\documentclass[a4paper,USenglish,cleveref,autoref]{lipics-v2021}
\usepackage{booktabs}
\usepackage{xcolor}

\usepackage{multirow}
\usepackage{xurl}
\usepackage{tikz}
\usetikzlibrary{positioning,decorations.pathreplacing}
\hideLIPIcs
\nolinenumbers

\title{Taxonomic Classification with Tag Arrays}
\titlerunning{Taxonomic Classification with Tag Arrays}

\author{Travis Gagie}{CeBiB -- Center for Biotechnology and Bioengineering, Chile}{travis.gagie@gmail.com}{https://orcid.org/0000-0003-3689-327X}{}
\author{Gonzalo Navarro}{Department of Computer Science, University of Chile, Santiago, Chile \and CeBiB -- Center for Biotechnology and Bioengineering, Chile}{gnavarro@dcc.uchile.cl}{https://orcid.org/0000-0002-2286-741X}{}

\authorrunning{T. Gagie and G. Navarro}
\Copyright{Travis Gagie and Gonzalo Navarro}


\acknowledgements{Many thanks to Ben Langmead for recommending classification with MEMs, and to Omar Ahmed, Victor Alfaro and Christina Boucher.  The first author thanks Dr.\ Curtis Marcoux for leading the team that treated him for the past two years, Dr.\ Amy Trottier for convincing him that bioinformatics is actually useful, and the staff of the 4th and 8th floors of the Victoria General Hospital in Halifax.  The ideas in this paper are the authors'; the code, the experiments and the text were developed by Claude Opus 5.5, an AI model by Anthropic, under the authors' direction.}

\EventEditors{John Q. Open and Joan R. Access}
\EventNoEds{2}
\EventLongTitle{42nd Conference on Very Important Topics (CVIT 2016)}
\EventShortTitle{CVIT 2016}
\EventAcronym{CVIT}
\EventYear{2016}
\EventDate{December 24--27, 2016}
\EventLocation{Little Whinging, United Kingdom}
\EventLogo{}
\SeriesVolume{42}
\ArticleNo{23}

\begin{document}
\maketitle
\begin{abstract}
Taxonomic classifiers based on exact matches differ in which matches they use and in how each match splits its credit among the taxa containing it: Kraken~2 gives a $k$-mer's credit to the lowest common ancestor of those taxa, SPUMONI~2 labels each position of the read, Tagger gives each maximal exact match (MEM) one label, and Cliffy splits a match's credit evenly over a roughly logarithmic sample of the taxa.
We compare these schemes on the same long MEMs, mainly on short 16S rRNA reads classified to genera with SILVA, where credit proportional to the number of occurrences of the MEM in each taxon is the most accurate, by 1.0 to 6.8 percentage points.
We compute long MEMs with Boyer--Moore--Li on a run-length compressed suffix array, and count their occurrences in each taxon with a complete, run-length compressed tag array, in time proportional to the number of tag runs in the MEM's BWT interval, without locating any occurrence.
We bound the number of tag runs in terms of the BWT runs and of how often suffixes with long common prefixes have different tags, and show that tags constant on $d$ blocks of the text have a macro scheme with at most $r + 2d$ phrases.
On SILVA our index takes 1.44\,GB and classifies a read in 66\,$\mu$s with one thread at 93.8\% accuracy; on the same reads and server, it is 14 times smaller than Cliffy's undigested index, and 4.4 times faster and 8 percentage points more accurate than Cliffy's LCA query on it. Over minimizer digests it takes 0.49\,GB and reaches 92.5\%; on the server, 92.1\% against 84.2\% for Cliffy's LCA query on its digested index.
On nanopore reads from whole genomes, where almost every long MEM occurs in one species, we instead use each read's longest MEMs, found in one pass of a variant of Boyer--Moore--Li with a $k$-mer filter, which assigns 99.6\% of the reads to their species with no length threshold.
\end{abstract}
\clearpage

\section{Introduction}
\label{sec:intro}

A sequencing experiment on a microbial sample produces millions of \emph{reads}, strings copied from random places in the DNA of the organisms in the sample: a few hundred bases long and accurate with short-read technologies, or many thousands of bases long and with more errors with long-read technologies such as Oxford Nanopore (ONT).
A read can come from either strand, so it may match the database only as a \emph{reverse complement}.
Taxonomic classification assigns each read to a node of a taxonomy, such as a species or genus, by comparing it with a reference database of sequences whose taxa are known.
Kraken~\cite{WS14} made this fast by precomputing, for every $k$-mer in the database, the lowest common ancestor (LCA) in the taxonomy of the genomes containing it, and Kraken~2~\cite{WLL19} made the database much smaller by storing minimizers in a compact probabilistic hash table.
As databases grow, however, fewer $k$-mers are specific to a single taxon, so $k$-mer LCAs move up the taxonomy and species-level resolution is lost, both for whole genomes~\cite{NKPT18} and for the 16S rRNA gene and other marker genes~\cite{CLMP24}.
In work in preparation, Cheng et al.~\cite{CALL} report that maximal exact matches (MEMs), which are as long as the match allows (Section~\ref{sec:schemes}), give better classifications than any single choice of $k$.

Meanwhile, more and more classifiers are built with machine learning instead: from models of $k$-mer frequencies~\cite{RSE+20,MAHK22,LKYH23,AHL+24} to deep networks trained on reads~\cite{LBL+20,MKK+22} and language models pretrained on genomes~\cite{JZLD21,DLT+25}.
Such models are easy to apply to new data, but their answers are hard to explain and depend on their training data in ways that are hard to inspect.
Classifiers based on pattern matching are exact and their answers can be traced to specific matches, but they remain the natural choice only if they keep pace with growing databases, in accuracy and in the time and space they need.
This paper asks what that requires.

A classifier based on exact matches can use $k$-mers, the phrases of a greedy parse of the read, all the MEMs of the read, or only its long MEMs.
It must also decide how a match that occurs in several taxa splits its \emph{credit}, which we take to be the match's length, among those taxa.
Existing tools differ in both choices (Section~\ref{sec:schemes}).
Kraken~2 gives all the credit of a $k$-mer to the LCA of the taxa containing it.
SPUMONI~2~\cite{ARG+23} gives each position of a read the label of one occurrence of its longest match.
Tagger~\cite{DAF+25} gives each MEM of at least a given length the label of one of its occurrences.
Cliffy~\cite{ABL25} splits the credit of each phrase evenly over a sample of the taxa containing it, defined by records in the order of the taxa and of roughly logarithmic size, or over the range of taxa between the leftmost and rightmost ones containing it.
We also consider splitting the credit evenly among all the taxa containing the match, and splitting it in proportion to the number of occurrences of the match in each taxon.

Our contributions are the following.
\begin{enumerate}
\item We compare these choices on the same matches, in two regimes that stress classification differently (Sections~\ref{sec:regimes} and~\ref{sec:accuracy}): real ONT reads from the ZymoBIOMICS mock community, classified to species against whole reference genomes, as in the SPUMONI~2 experiments~\cite{ARG+23}; and short simulated reads from the conserved 16S rRNA gene, classified to genera against the SILVA database~\cite{QPY+13}, which is small, very repetitive and unevenly populated, as in the Cliffy experiments~\cite{ABL25}.
Long MEMs make classification accurate in both regimes.
On the Zymo reads, almost every long MEM occurs in one species, so the split hardly matters and the choice of matches is what counts; on SILVA, where it does matter, credit proportional to the number of occurrences is clearly the most accurate split.
Proportional credit gives each taxon the expected weight of a single label drawn uniformly from the occurrences of the match, without the variance of such a draw.
It also favors taxa with much sequence in the database, so it can be expected to help only when the database resembles the sample (Section~\ref{sec:schemes}).
\item We show that long MEMs and their counts can be computed fast in compressed space (Sections~\ref{sec:bml} and~\ref{sec:tags}).
Long MEMs are found with Boyer--Moore--Li (BML)~\cite{Gag24,Li24}, which skips most of the parts of a read that cannot be in a long MEM, on a run-length compressed suffix array (RLCSA)~\cite{MNSV10,BGM+25}.
The counts come from a complete, run-length compressed \emph{tag array}~\cite{BGG+24} that stores the taxon of every suffix of the reference in suffix-array order, where Tagger stores only a sample of the tags: the occurrences of a MEM in each taxon are counted by scanning the tag runs that overlap the MEM's BWT interval, without locating any occurrence.
The tag array is built from the BWT in space proportional to the number of BWT runs plus the number of tag runs, and the sequence of its runs can be grammar-compressed, because LF maps the tags under each BWT run to a copy of themselves.
On SILVA the resulting index (Section~\ref{sec:tradeoffs}) is 14 times smaller than Cliffy's undigested index and faster and more accurate than Cliffy's LCA query on the same reads and server, builds in minutes on a desktop computer where Cliffy's took almost three days on the server, and is more accurate than Kraken~2 and Tagger on the same machine and reads.
\end{enumerate}
Our implementation, KATKA, is available at \url{https://github.com/TravisGagie/KATKA}; we describe it only as far as needed to support these results.

\section{Matches and credit}
\label{sec:schemes}
\label{sec:background}

\subparagraph*{Matches.}
Let $S$ be the reference text, consisting of the reference sequences and their reverse complements (reversed, with A and T, and C and G, swapped), grouped by taxon (Section~\ref{sec:index}), and let $R$ be a read.
Following~\cite{Gag24}, a \emph{MEM} of $R$ is a substring $R[a..e - 1]$ that occurs in $S$ but cannot be extended to the left or to the right and still occur.
Maximality is thus with respect to intervals of the read, not to pairs of occurrences as is usual in stringology, and these MEMs are the super-maximal exact matches (SMEMs) of Li~\cite{Li12} and of Tagger.
Cliffy instead cuts the read into consecutive \emph{Ziv--Merhav phrases} (ZMPs) by searching backward from the end of the read for as long as the match occurs, and starting a new search at the character where the previous one failed.
This mirrors the cross parsing of Ziv and Merhav~\cite{ZM93}, later used as Relative Lempel--Ziv (RLZ) parsing~\cite{KPZ10}.
Each ZMP $R[j..f - 1]$ is the longest suffix of $R[0..f - 1]$ that occurs in $S$, so it cannot be extended to the left and is a prefix of the MEM starting at $j$; ZMPs are generally shorter than the MEMs that contain them.
A $k$-mer is a substring of fixed length $k$, and Kraken~2 uses minimizers: from each $k$-mer it keeps the substring of a fixed shorter length with the smallest hash.
Short matches occur in many taxa, so they say little about where a read came from; with Boyer--Moore--Li we use only the MEMs of length at least a threshold $L$.

\subparagraph*{Credit.}
Let $M$ be a match, let $c_g$ be the number of occurrences of $M$ in the reference sequences of taxon $g$, counting both strands, let $D = \{g : c_g > 0\}$ and let $C = \sum_g c_g$.
A \emph{credit scheme} splits the credit $|M|$ into weights $w_v \geq 0$ for the nodes $v$ of the taxonomy, with $\sum_v w_v = |M|$; all schemes except Kraken~2's give weight only to the leaves, which are the taxa we classify to.
The read is assigned to the taxon with the largest total weight over all its matches, except with Kraken~2's scheme.
Table~\ref{tab:schemes} lists the schemes we consider, with the $d$ taxa numbered $1, \ldots, d$ in the left-to-right order of the leaves of the taxonomy, so that the taxa under any node form an interval, and with $g_\ell = \min D$ and $g_r = \max D$.
\begin{description}
\item[LCA.] Kraken~2 gives $|M|$ to the LCA of $D$, and assigns the read to the leaf whose root-to-leaf path collects the most weight.
Cliffy's LCA query instead splits the credit evenly over the \emph{range} of taxa from $g_\ell$ to $g_r$, giving $|M|/(g_r - g_\ell + 1)$ to each, including taxa that do not contain $M$; all of them lie under the LCA.
\item[One label per position.] For each position of the read, SPUMONI~2 computes the length of the longest match starting there (the matching statistics, approximated), with the r-index~\cite{GNP20}, and for classification into several classes the class of one occurrence of that match, read from document-array samples at the boundaries of BWT runs~\cite{ARG+23}; the read goes to the class with the most positions.
\item[One label per MEM.] Tagger gives $|M|$ to the taxon of the last occurrence of $M$ in BWT order, which it maintains during the search much as the r-index maintains its \emph{toehold}, the position of one occurrence of the current match.
That occurrence is more likely to lie in a taxon with many occurrences, so we model this scheme as giving $|M|$ to one taxon $g$ drawn with probability $c_g / C$, although Tagger's label is deterministic.
\item[Sample.] Cliffy's approximate document listing gives $|M|/|D'|$ to each taxon of a subset $D' \subseteq D$.
At the boundaries of BWT runs, Cliffy stores a \emph{document array profile}~\cite{ARBL23}, which gives for each taxon the length of the longest common prefix of the suffix at that position with any suffix of that taxon, and keeps only its \emph{cliffs}: the entries larger than every entry to their left, or to their right, in the order of the taxa.
These are left-to-right and right-to-left maxima, or records (B.~Langmead, personal communication), and a random permutation of $d$ distinct values has on average $H_d = \ln d + O(1)$ left-to-right maxima~\cite{Ren62}, so there are often about logarithmically many cliffs.
$D'$ consists of the taxa whose cliffs are at least $|M|$, so it is determined by the order of the taxa, not drawn at random or weighted by counts, and it always contains $g_\ell$ and $g_r$.
\item[Even.] Each taxon in $D$ gets $|M|/|D|$, as Cliffy's approximate document listing would if it listed $D$ exactly.
\item[Proportional.] Each taxon in $D$ gets $|M| \cdot c_g / C$.
\end{description}

\begin{table}[t]
\centering
\caption{Matches and credit schemes. $M$ is a match, $c_g$ is the number of its occurrences in taxon $g$, $C = \sum_g c_g$, $D$ is the set of taxa containing it, and $g_\ell$ and $g_r$ are the leftmost and rightmost of them in the order of the leaves of the taxonomy.}
\label{tab:schemes}
\resizebox{\textwidth}{!}{\begin{tabular}{llll}
\toprule
classifier & matches & credit scheme & weights \\
\midrule
Kraken~2~\cite{WLL19} & $k$-mers (minimizers) & LCA & $|M|$ to $\mathrm{LCA}(D)$ \\
SPUMONI~2~\cite{ARG+23} & longest match at each position & one label per position & -- \\
Tagger~\cite{DAF+25} & MEMs of length $\geq L$ & one label per MEM & $|M|$ to one $g \in D$ \\
Cliffy~\cite{ABL25}, LCA query & ZMPs & range & $|M| / (g_r - g_\ell + 1)$ to each of $g_\ell..g_r$ \\
Cliffy~\cite{ABL25}, approx.\ listing & ZMPs & sample & $|M|/|D'|$ to each of $D' \subseteq D$ \\
-- & MEMs of length $\geq L$ & even & $|M|/|D|$ to each of $D$ \\
ours & MEMs of length $\geq L$ & proportional & $|M| \cdot c_g / C$ to each $g \in D$ \\
\bottomrule
\end{tabular}}
\end{table}

\subparagraph*{Why proportional credit.}
If one label is drawn uniformly at random from the occurrences of $M$, taxon $g$ gets it with probability $c_g / C$, so proportional credit gives each taxon the expected weight of such a draw.
Summed over the matches of a read, proportional credit keeps the expected total weight of each taxon under one label per match, without the variance of the draws, which matters most for reads with few long matches.
Similarly, even credit gives each taxon the expected weight of a label drawn uniformly from $D$.
A sample like Cliffy's is not a uniform subset of $D$, since it always contains $g_\ell$ and $g_r$, so it only roughly approximates even credit.
For a single match, $|M| \cdot c_g / C$ is also $|M|$ times the posterior probability that a string containing $M$ was drawn from taxon $g$, if it was drawn from a uniformly random position of the database, ignoring the boundaries of the sequences and overlapping occurrences of $M$.
The implicit prior is thus proportional to the amount of sequence of each taxon in the database, and summing over the matches of a read is a heuristic; it does not give the posterior given all the matches.
This prior can help when the taxa that are abundant in the sample also have much sequence in the database, and hurt when they do not; then normalizing each count by the size of its taxon, or correcting towards a known or estimated composition of the sample, may be better.

Even credit needs the set $D$ itself, not just $g_\ell$ and $g_r$, and proportional credit needs the counts, which matter when a match occurs much more often in some taxa than in others (Figure~\ref{fig:lca}).

\begin{figure}[t]
\centering
\begin{tikzpicture}[level distance=6mm, every node/.style={inner sep=1.2pt},
  level 1/.style={sibling distance=40mm}, level 2/.style={sibling distance=20mm}, level 3/.style={sibling distance=10mm}]
\node[circle,fill] {}
  child {node[circle,fill] {}
    child {node[circle,fill] {} child {node (g1) {$g_1$}} child {node[draw,thick] (g2) {$g_2$}}}
    child {node[circle,fill] (lca) {} child {node (g3) {$g_3$}} child {node (g4) {$g_4$}}}}
  child {node[circle,fill] {}
    child {node[circle,fill] {} child {node[draw,thick] (g5) {$g_5$}} child {node (g6) {$g_6$}}}
    child {node[circle,fill] {} child {node (g7) {$g_7$}} child {node (g8) {$g_8$}}}};
\node[below=1mm of g2, font=\scriptsize, align=center] {$c = 9$\\even $|M|/2$\\prop.\ $0.9|M|$};
\node[below=1mm of g5, font=\scriptsize, align=center] {$c = 1$\\even $|M|/2$\\prop.\ $0.1|M|$};
\draw[decorate, decoration={brace, mirror, amplitude=4pt}] ([yshift=-15mm]g2.west) -- ([yshift=-15mm]g5.east) node[midway, below=4pt, font=\scriptsize] {range (Cliffy's LCA query): $|M|/4$ to each of $g_2, \ldots, g_5$};
\end{tikzpicture}
\caption{A match $M$ that occurs 9 times in taxon $g_2$ and once in $g_5$ (boxed).
The LCA of $g_2$ and $g_5$ is the root, so with Kraken~2's scheme the match gives weight only to the root.
Cliffy's LCA query splits the credit $|M|$ evenly over the taxa from the leftmost to the rightmost one containing the match, including $g_3$ and $g_4$, which do not contain it.
Even credit splits the credit between $g_2$ and $g_5$, and proportional credit gives $g_2$ nine times as much weight as $g_5$.}
\label{fig:lca}
\end{figure}
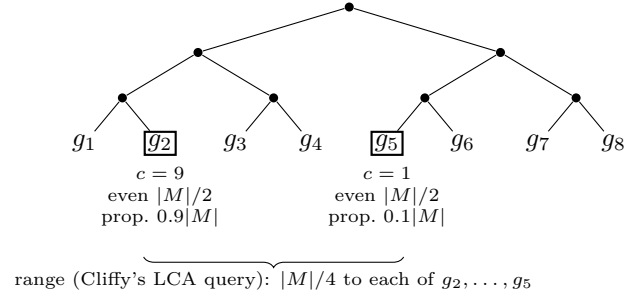

Other r-index classifiers, such as SPUMONI~2 and col-BWT~\cite{BSL26}, mostly decide whether a read comes from a target pangenome; Cobas, M\"akinen and Rossi~\cite{CMR20} report document frequencies from a grammar-compressed document array; and Omri et al.~\cite{OSM+26} classify genomic sequences by compression.

\section{Two regimes}
\label{sec:regimes}
\label{sec:setup}

Our two experiments stress different choices.
On the Zymo reads, almost every long MEM occurs in one species, so what matters is which matches to use; on SILVA, many long MEMs occur in several genera, so what matters is how to split their credit.
We take them in that order, which is also the order of the SPUMONI~2 and Cliffy experiments that they reproduce.

\subparagraph*{Long reads from whole genomes: Zymo.}
SPUMONI~2's experiments~\cite{ARG+23} use ONT reads of a ZymoBIOMICS \emph{mock community}, a laboratory mixture of known bacteria and yeasts, with a reference of 14\,280 RefSeq sequences of seven of its bacterial species, from 267 for \emph{B.~subtilis} to 7\,914 for \emph{E.~coli}, and minimap2 alignments to label each read with a species or as yeast.
N.~Brown made the reference, the reads and their labels available.\footnote{\url{https://zenodo.org/records/23251764}}
There are 581\,802 bacterial reads with a species label, which we classify whole.
SPUMONI~2's Table~2 instead takes the first 180 to 720 bases of each of these reads and of 10\,245 yeast reads, and asks whether each prefix comes from the bacterial reference; the \emph{sensitivity} is the fraction of bacterial reads that are accepted, and the \emph{specificity} the fraction of yeast reads that are rejected.
Here a taxon contains many complete genomes, which are long and close to each other, the reads are long but have frequent errors, especially in homopolymers, and some reference sequences contain long low-complexity regions, such as the short repeats at the ends of yeast chromosomes (telomeres).
We tag each suffix of $S$ with its species.
Tagging each suffix with its genome instead, which would make $S$ a collection of documents in the usual sense, gives a tag array with too many runs (Section~\ref{sec:tags}).
We homopolymer-compress the reference and the reads (Section~\ref{sec:bml}), after which $S$ has $n = 39.3$ billion characters and its BWT has $r = 338$ million runs, against $n = 53.7$ billion and $r = 437$ million without compression.
We count a read as correct if its species alone has the largest weight.

\subparagraph*{Short reads from a conserved gene: SILVA.}
Many studies of bacterial and archaeal communities sequence only copies of the 16S ribosomal RNA gene, which is about 1\,500 bases long and present in every bacterium and archaeon.
An experiment copies one or two of its nine \emph{hypervariable regions}, V1 to V9, which differ between taxa, and sequences each copy from both ends, giving two reads, or \emph{mates}; the reads are short and the regions similar across related taxa, so 16S reads are usually classified only to genera~\cite{JSL+19}.
We reproduced Cliffy's SILVA experiment~\cite{ABL25} as closely as we could from its authors' scripts.
We used SILVA 138.1 SSU NR99~\cite{QPY+13}, grouped its sequences by genus with their script, and kept 9\,118 genera containing 402\,378 sequences, in the left-to-right order of the leaves of the SILVA taxonomy.
The database is much smaller than the Zymo reference, but the genera are very unevenly represented, and sequences of different genera often share long substrings, which makes classification with exact matches hard and the database very repetitive.
With both strands, $S$ has $n = 1.17$ billion characters and its BWT has $r = 68.3$ million runs.
We simulated 10 million pairs of 250-base reads for each of the regions V1--V2, V3--V4, V4 and V4--V5 with the authors' simulator, MicrobeMixer, using the sequencing-error simulator ART~\cite{HLMM12} and the genus abundances of an aquatic community, and measured accuracy on every 50th pair.
Each mate is classified separately and the weights of the two mates are added; a read pair is \emph{strictly} correct if its true genus alone has the largest weight; ties and unclassified pairs count as wrong.
Cliffy's scoring also counts ties for the largest weight as correct, and pairs without any match unless their genus is ``uncultured'' (a placeholder for sequences not assigned to a known genus).
We report the strict accuracy unless stated otherwise, as the mean over the four regions.

\subparagraph*{Machines.}
We used a seven-year-old desktop (Intel Core i7-8700, 6 cores, 16\,GB of RAM) for SILVA and a two-year-old laptop (Intel Core i5-1335U, 12 threads, 14\,GB of RAM available) for Zymo.
Cliffy's indexes do not fit on either, so the second author built and ran them, and our default configuration, on our reads on a server with 256\,GB of RAM (Appendix~\ref{app:details}); our numbers for SPUMONI~2 and minimap2 are from their papers.

\section{Which matches and which credit}
\label{sec:accuracy}

\subparagraph*{Which matches: Zymo.}
On the Zymo reads, $L = 30$ homopolymer-compressed characters (about 41 bases) is the best fixed threshold (99.43\% with proportional credit), since ONT reads have few error-free stretches much longer.
Even so, 0.4\% of the reads have no MEM that long and stay unclassified.
Using only the longest MEMs of each read, measured in compressed characters (Section~\ref{sec:bml}), every read is classified, and proportional credit reaches 99.61\%, the best of all our configurations; on the 720-base prefixes, species accuracy is 97.88\% with proportional credit and 97.44\% with even credit.
How the credit is split hardly matters: the 89.6 million long MEMs of the whole reads occur in 1.01 species on average, and all the schemes of Table~\ref{tab:credit} are within 0.16 percentage points (\emph{points}, from now on) of each other: the sample reaches 99.58\%, even credit and the range 99.54\%, and proportional credit 99.43\%, and normalizing each count by the number of reference sequences of its species, or by their total length, gives 99.42\% and 99.43\%.

In the prefix test of SPUMONI~2's Table~2 with the first 720 bases of each read, we have 97.24\% sensitivity and 97.98\% specificity without homopolymer compression (96.88\% and 97.90\% with it; Section~\ref{sec:bml}), against 95.32\% and 97.06\% for SPUMONI~2 with minimizer digests, 97.62\% and 97.38\% for SPUMONI~2 without them, and 97.90\% and 97.85\% for minimap2.
With the longest MEMs, thresholding their length gives a whole sensitivity--specificity curve from one run: with homopolymer compression, 97.32\% and 96.33\% at 28 compressed characters, and 96.88\% and 97.90\% at 30.
The longest MEMs of yeast prefixes have a median of 23 compressed characters, about $\log_4$ of $n$ times the prefix length, as for random strings, so the threshold can be set a few characters above that rather than tuned.

\subparagraph*{Which credit: SILVA.}
On SILVA, a long MEM of one of our reads occurs in 3.8 genera on average (Section~\ref{sec:tags}), so the split matters.
Long MEMs are still the right matches.
With the range scheme, all the MEMs of each read were 0.4 to 1.3 points more accurate than ZMPs, both on the sequences and on their \emph{minimizer digests}~\cite{RHH+04,ARG+23} (Appendix~\ref{app:matches}).
With even credit, strict accuracy is 90.56\% with $L = 15$ and 90.28\% with $L = 30$, drops by about 2 points at $L = 50$ and collapses for larger $L$, as the 250-base reads start to have no MEM that long (Table~\ref{tab:L}); with proportional credit it is 93.74\% with $L = 15$, 93.77\% with $L = 30$ and 93.48\% with $L = 40$.
Discarding short MEMs also makes classification much faster (Section~\ref{sec:tradeoffs}).

Table~\ref{tab:credit} compares the credit schemes on the same long MEMs, so that only the split differs.
For one label per MEM and for the sample, we simulate the schemes of Tagger and Cliffy from the counts: one label drawn with probability $c_g/C$, and the left-to-right and right-to-left maxima of random values given to the taxa in $D$, which form a sample containing $g_\ell$ and $g_r$ with $2 H_{|D|} - 1$ taxa on average.
The tools themselves give answers that are not random, and use different matches, so this comparison measures the credit schemes, not the tools.
One label drawn by counts is 1.6 points less accurate than proportional credit, which has the same expected weights, so this is roughly what the variance of one label per MEM costs; it is still 1.9 points more accurate than even credit.
The sample is about as accurate as even credit, and both are 3.3 to 3.5 points less accurate than proportional credit.
Proportional credit also almost eliminates ties, and it is 1 point more accurate than normalizing each count by twice the number of sequences in the genus, which gives each genus credit in proportion to the fraction of its sequences containing the match.
Even credit is more accurate than the range: with $L = 15$, by 3.5 points without digestion and 3.6 points with it (Appendix~\ref{app:matches}).
On the same machine and reads, Kraken~2 and Tagger are less accurate than long MEMs with proportional credit (Table~\ref{tab:others}).

\begin{table}[t]
\centering
\caption{Accuracy of the credit schemes on the same long MEMs ($L = 30$): species level on the Zymo reads (whole homopolymer-compressed reads, low-complexity MEMs ignored), and genus level on SILVA (strict, mean over the four regions, 200\,000 read pairs per region). One label per MEM and the sample simulate Tagger's and Cliffy's schemes from the counts.}
\label{tab:credit}
\small
\begin{tabular}{lcc}
\toprule
credit scheme & Zymo (species) & SILVA (genus) \\
\midrule
range (Cliffy's LCA query) & 99.54\% & 87.0\% \\
one label per MEM (drawn by counts) & 99.44\% & 92.18\% \\
sample (simulated cliffs) & 99.58\% & 90.47\% \\
even & 99.54\% & 90.28\% \\
proportional, normalized by taxon size & 99.42\% & 92.76\% \\
proportional & 99.43\% & 93.77\% \\
\bottomrule
\end{tabular}
\end{table}

\subparagraph*{Composition.}
MicrobeMixer chooses the genus of each simulated read according to the abundances of an aquatic community, and only then a sequence of that genus, so our SILVA reads are not drawn in proportion to how well each genus is represented in SILVA.
Proportional credit is still the most accurate, and its 1-point advantage over normalized credit suggests that its prior is an advantage here, presumably because genera that are abundant in the community are also well represented in SILVA.

\section{Finding long MEMs}
\label{sec:bml}
\label{sec:index}

\subparagraph*{The index.}
For $i = 1, \ldots, d$, let $G_i$ be the concatenation of the reference sequences of taxon $i$, separated by a character \texttt{X} that matches nothing, with the taxa numbered in the left-to-right order of the leaves of the taxonomy.
We index the text $S = G_1 \mathtt{X}\, \mathrm{rc}(G_1) \mathtt{X} \cdots G_d \mathtt{X}\, \mathrm{rc}(G_d) \mathtt{\$}$ of length $n$, where $\mathrm{rc}$ denotes the reverse complement and \texttt{\$} is a unique terminator, smaller than every other character; no match crosses a separator.
We call $G_i \mathtt{X}\, \mathrm{rc}(G_i) \mathtt{X}$ the \emph{block} of taxon $i$ (for $i = d$, the last character is the terminator).
Let $\mathrm{SA}$ be the suffix array of $S$, let the BWT be the string whose $i$th character is $S[\mathrm{SA}[i] - 1]$, or \texttt{\$} if $\mathrm{SA}[i] = 1$, and let $r$ be the number of maximal runs of equal characters in the BWT.
The \emph{tag array} $T[1..n]$ gives, for each $i$, the taxon $T[i]$ of the block containing $\mathrm{SA}[i]$.
LF maps the BWT position of a suffix to that of the suffix one character longer, $\phi$ maps $\mathrm{SA}[i]$ to $\mathrm{SA}[i - 1]$, and the \emph{BWT interval} of a string is the interval of BWT positions of the suffixes it prefixes; \emph{backward search} computes the BWT interval of $cP$ from that of $P$.
Since $S$ contains both strands, a read matches $S$ whichever strand it came from, and extending a match to the right by a base $c$ is the same as extending its reverse complement to the left by the complement of $c$, so backward search alone can extend matches in both directions, although switching direction means searching again for the reverse complement of the whole match.
We store the BWT of $S$ as an RLCSA in the form of Brown et al.~\cite{BGM+25}, which supports LF and backward search directly: for each character $c$, the sorted list of the starting positions of the runs of $c$ in the BWT, with the cumulative lengths of those runs, so that a backward step by $c$ is a predecessor search in the list of $c$, whose cost does not depend on the size of the alphabet.
As Bowtie~2~\cite{LS12} and Cliffy do, we also store the BWT interval of every string of length $t = 10$, so that a search starts with one table lookup; if the string does not occur, we fall back to stepping base by base, so we still know exactly where a match fails.

\subparagraph*{Boyer--Moore--Li.}
BML~\cite{Gag24,Li24} finds the MEMs of $R$ of length at least $L$, writing $R[a..e - 1]$ for the substring from position $a$ to position $e - 1$.
Starting with $x = 0$, it checks whether a MEM of length at least $L$ starts in $R[x..]$ by searching backward for the window $R[x..x + L - 1]$, starting from its right end.
If the step that prepends $R[j - 1]$ fails, so that $R[j - 1..x + L - 1]$ does not occur, then no MEM of length at least $L$ starts at a position from $x$ to $j - 1$, and BML jumps to $x = j$, skipping the rest of the window as in Boyer--Moore pattern matching.
If the whole window occurs, BML extends the match to the left as far as it can, to some position $a$, then to the right, by searching backward for the reverse complement of $R[a..x + L - 1]$ and continuing to some position $e - 1$, reports the MEM $R[a..e - 1]$ and continues with $x = e - L + 1$, the first window that ends after the MEM.
Figure~\ref{fig:bml} shows an example.
Gagie~\cite{Gag24} gives the theory, the full version~\cite{Gag24arxiv} gives a practical algorithm, and Li implemented BML in ropebwt3~\cite{Li24}.
At the end of the forward extension we have the BWT interval of the MEM's reverse complement, whose occurrences are in the same taxa as the MEM's, which is all the tag array needs.

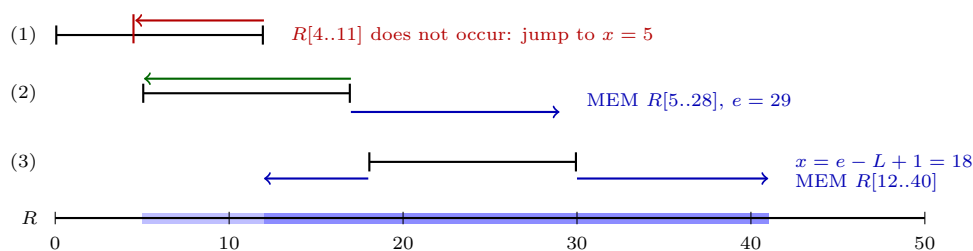
\begin{figure}[t]
\centering
\begin{tikzpicture}[x=2.3mm, y=5.5mm, every node/.style={font=\scriptsize}]
\fill[blue!25] (5,-0.13) rectangle (29,0.13);
\fill[blue!45] (12,-0.13) rectangle (41,0.13);
\draw[thick] (0,0) -- (50,0); \node[anchor=east] at (-0.5,0) {$R$};
\foreach \i in {0,10,20,30,40,50} { \draw (\i,-0.15) -- (\i,0.15); \node[below=2pt] at (\i,-0.1) {\i}; }
\node[anchor=east] at (-0.5,4.4) {(1)};
\draw[|-|, thick] (0,4.4) -- (12,4.4);
\draw[->, red!70!black, thick] (12,4.75) -- (4.6,4.75); \draw[red!70!black, thick] (4.5,4.2) -- (4.5,4.9);
\node[red!70!black, anchor=west] at (13,4.4) {$R[4..11]$ does not occur: jump to $x = 5$};
\node[anchor=east] at (-0.5,3.0) {(2)};
\draw[|-|, thick] (5,3.0) -- (17,3.0);
\draw[->, green!40!black, thick] (17,3.35) -- (5.1,3.35);
\draw[->, blue!70!black, thick] (17,2.55) -- (29,2.55);
\node[blue!70!black, anchor=west] at (30,2.8) {MEM $R[5..28]$, $e = 29$};
\node[anchor=east] at (-0.5,1.35) {(3)};
\draw[|-|, thick] (18,1.35) -- (30,1.35);
\draw[->, blue!70!black, thick] (18,0.95) -- (12,0.95);
\draw[->, blue!70!black, thick] (30,0.95) -- (41,0.95);
\node[blue!70!black, anchor=west] at (42,1.35) {$x = e - L + 1 = 18$};
\node[blue!70!black, anchor=west] at (42,0.85) {MEM $R[12..40]$};
\end{tikzpicture}
\caption{Boyer--Moore--Li with $L = 12$ on a read $R$.
(1)~The search for the window $R[0..11]$ fails at the step that prepends $R[4]$, so no MEM of length at least $L$ starts at positions 0 to 4, and BML jumps to $x = 5$.
(2)~The window $R[5..16]$ occurs; it cannot be extended to the left, since $R[4..16]$ contains $R[4..11]$, so BML extends it to the right and reports the MEM $R[5..28]$.
(3)~The next window is $R[18..29]$; BML extends it in both directions, to the MEM $R[12..40]$, and continues with $x = 30$.}
\label{fig:bml}
\end{figure}

\subparagraph*{Long, noisy reads.}
ONT reads have many errors in \emph{homopolymers}, runs of a repeated base, and these errors break long matches.
We therefore \emph{homopolymer-compress} the reference and each read, replacing every run of equal bases by a single base, and measure $L$ in compressed characters ($L = 30$ is about 41 bases on the Zymo reads); spans are mapped back to the read's coordinates for scoring.
Homopolymer compression also turns some low-complexity sequences into short periodic ones: yeast telomeric repeats become nearly $(\mathtt{CA})^*$, and AT-rich regions of mitochondrial DNA nearly $(\mathtt{TA})^*$.
Long MEMs in such regions occur in many places, and mostly in the reference genomes of a few species, so they pull reads towards those species.
We therefore ignore every MEM with fewer than $\tau$ distinct 3-mers, computed on the compressed sequence; we use $\tau = 14$, out of the 36 3-mers without repeated consecutive bases.
On the 720-base Zymo prefixes, this raises the specificity of the test of SPUMONI~2's Table~2 from 89.18\% to 97.90\% at the cost of 0.03 points of sensitivity; without homopolymer compression (and $\tau = 16$ of 64), the test gives 97.24\% and 97.98\%, so compression mainly makes the index smaller ($r = 338$ instead of 437 million).

\subparagraph*{Longest MEMs.}
BML can find the longest MEMs of a read in one pass, without a threshold: start with $L = 1$ and, whenever a MEM longer than $L$ is found, set $L$ to its length.
Raising $L$ only stops the search for MEMs shorter than one already found, and a MEM as long as the current $L$ is still found, so the longest MEMs, including ties, are among those reported; as $L$ grows, BML skips more of the read.
KeBaB~\cite{BDZ+25} first cuts the read into \emph{pseudo-MEMs}, the maximal substrings of length at least $k$ all of whose $k$-mers pass a Bloom filter of the reference's $k$-mers, so that every MEM lies in a pseudo-MEM.
Searching only the longest pseudo-MEMs, as in KeBaB's pipeline, can miss the longest MEM, since a long pseudo-MEM need not contain one; instead we search the pseudo-MEMs from longest to shortest and stop at the first one shorter than the current $L$, which loses nothing.
On the 581\,802 whole Zymo reads, this finds exactly the same longest MEMs as without the filter, which takes 0.135\,GB, and classification takes 118\,s with 12 threads on the laptop, against 211\,s without the filter and about 780\,s with $L = 30$; on SILVA, whose short, accurate reads have almost all their $k$-mers in the database, the filter saved little time (Appendix~\ref{app:alternatives}).

\section{Counting with complete tag arrays}
\label{sec:tags}
\label{sec:freq}

\subparagraph*{Tag arrays.}
A tag array associates a tag with each position of the BWT~\cite{BGG+24,ESP25,ESP26,OO26}, and Wheeler maps~\cite{BGG+24} bound its space in terms of the number of runs in the tags.
Here the tag array is $T$, in which the tag of a BWT position is the taxon of the corresponding suffix of $S$.
If a document is the concatenation of the sequences of one taxon, then this tag array is the document array of Muthukrishnan~\cite{Mut02}; other tag arrays, such as those of Eskandar et al.~\cite{ESP25,ESP26}, whose tags are vertices of a pangenome graph, and of Olbrich and Ohlebusch~\cite{OO26}, whose tags are columns of a multiple alignment, are not document arrays.
Tagger~\cite{DAF+25} stores only the tags at the ends of BWT runs, from which the tags of other positions cannot in general be recovered.
To stress the difference, we call a tag array \emph{complete} if it gives the tag of every BWT position, whether it stores them run-length compressed or recovers them from a sample.
We run-length compress the tag array, storing the run starts with Elias--Fano coding and the tag of each run, and let $\rho$ be the number of runs in $T$, which we call \emph{tag runs}.

\subparagraph*{Contextual locality.}
The runs are few because the BWT sorts the suffixes of $S$, and with them the characters preceding them, by their \emph{contexts}, the strings that follow those characters.
A property has \emph{contextual locality}~\cite{Gag24tag,BGG+24,OO26} if suffixes that share long prefixes tend to have the same value of the property, their tag.
Olbrich and Ohlebusch~\cite{OO26} note that the tag array of such a property is then likely to be run-length compressible, for the same reason that the BWT of a repetitive text is; Lemma~\ref{lem:runs} makes this precise for tags that are constant on blocks of $S$, such as taxa.
Let $\Psi = \mathrm{LF}^{-1}$, so that $\mathrm{SA}[\Psi(i)] = \mathrm{SA}[i] + 1$, let a \emph{tag boundary} be a position $p > 1$ with $T[p - 1] \neq T[p]$, so that $\rho$ is one more than the number of tag boundaries, and let $\mathrm{lcp}(p)$ be the length of the longest common prefix of the suffixes starting at $\mathrm{SA}[p - 1]$ and $\mathrm{SA}[p]$.
A tag boundary $p$ with a large $\mathrm{lcp}(p)$ is a violation of contextual locality: two suffixes with a long common prefix but different tags.

\begin{lemma}
\label{lem:runs}
Let $S[1..n]$, over an alphabet of size $\sigma$ and ending with a unique terminator, be cut into $d$ blocks, each with a tag, let $T[i]$ be the tag of the block containing $\mathrm{SA}[i]$, and let $r$ be the number of runs in the BWT of $S$.
For any $\lambda \geq 0$, if $B_\lambda$ is the number of tag boundaries $p$ with $\mathrm{lcp}(p) > \lambda$, then $\rho \leq (r + \sigma + 2d)(\lambda + 1) + B_\lambda + 1$.
\end{lemma}

\begin{proof}
Call a position $p > 1$ a \emph{break} if $\Psi(p) \neq \Psi(p - 1) + 1$.
$\Psi$ maps the positions whose suffixes start with a character $c$, in order, to the positions of the occurrences of $c$ in the BWT, so a break is either a position where the first character of the suffixes changes, or one where two consecutive occurrences of $c$ in the BWT are in different runs; there are at most $r + \sigma$ breaks.
Let $p_0$ be a tag boundary with $\mathrm{lcp}(p_0) \leq \lambda$, and while $p_h$ is a tag boundary and not a break, let $p_{h + 1} = \Psi(p_h)$.
Then $\Psi(p_h - 1) = p_{h + 1} - 1$, so the suffixes at positions $p_{h + 1} - 1$ and $p_{h + 1}$ are those at $p_h - 1$ and $p_h$ without their first character, which is the same, and $\mathrm{lcp}(p_{h + 1}) = \mathrm{lcp}(p_h) - 1$.
Let $h$ be the largest index such that $p_h$ is a tag boundary.
A position with lcp $0$ is a break, so $h \leq \mathrm{lcp}(p_0) \leq \lambda$.
Either $p_h$ is a break, or $p_{h + 1}$ is not a tag boundary although $p_h$ is, which means that one of $\mathrm{SA}[p_h - 1]$ and $\mathrm{SA}[p_h]$ is the last position of a block; each of the $d$ last positions is $\mathrm{SA}[q]$ for one $q$, and then $p_h \in \{q, q + 1\}$, so there are at most $2d$ such positions $p_h$.
Since $p_0 = \mathrm{LF}^h(p_h)$, the pair $(p_h, h)$ determines $p_0$, so there are at most $(r + \sigma + 2d)(\lambda + 1)$ tag boundaries with lcp at most $\lambda$.
Adding the $B_\lambda$ other tag boundaries, and 1 because $\rho$ is one more than the number of tag boundaries, gives the bound.
\end{proof}

$B_\lambda$ itself measures violations of contextual locality, so the content of the lemma is that all the other tag boundaries number $O((r + \sigma + d)(\lambda + 1))$, independently of $n$: $\rho$ is small when $r$ and $d$ are small and, for some modest $\lambda$, few suffixes adjacent in suffix-array order share a prefix longer than $\lambda$ but have different tags.
In a reference with many genomes per species, we expect nearly every suffix to share a longer prefix with the corresponding suffixes of other genomes of its species than with any suffix of another species, so that species have contextual locality.
On the homopolymer-compressed Zymo reference, $\rho = 449$ million, only 1.33 times $r$.
If each genome is instead its own tag, the tag array is a document array, and corresponding suffixes of the genomes of a species, which are adjacent in suffix-array order and share long prefixes, have different tags, so $B_\lambda$ is large for any modest $\lambda$ and Lemma~\ref{lem:runs} gives no useful bound; in practice that tag array had so many runs that its construction ran out of memory.
On SILVA, $\rho = 359$ million, 5.3 times $r$, presumably because 16S sequences of different genera often share long substrings.

\subparagraph*{Counting.}
To count the occurrences of a MEM in each taxon, we scan the runs that overlap the MEM's BWT interval and add up the lengths of the runs of each taxon, clipped to the interval.
After a predecessor search for the first run, this takes time proportional to the number of runs in the interval rather than to the number of occurrences, and never locates an occurrence.
When counts are not needed, Muthukrishnan's document listing algorithm~\cite{Mut02} lists the taxa in an interval in time proportional to their number, with a range-minimum structure~\cite{FH11} over the array that points from each run to the previous run with the same tag; as in Sadakane's method~\cite{Sad07} and as Olbrich and Ohlebusch do~\cite{OO26}, that array need not be stored.
The usual alternative is to locate every occurrence with an r-index~\cite{GNP20}, from the toehold with $\phi$ queries, and map each to its taxon with a predecessor search over the taxon starts.
This also gives exact counts, but costs a $\phi$ query, with its cache misses, for every occurrence.
On SILVA the difference is large, because genera contain many near-identical sequences: a MEM of at least 30 bases in one of our reads occurs 70 times on average, in 3.8 genera; half of these MEMs occur only once and 5\% occur more than 200 times, yet their BWT intervals contain few tag runs.
With the same RLCSA and lookup table, classifying with an r-index takes 125\,$\mu$s per read, against 66\,$\mu$s with the tag array, for an index only 0.06\,GB smaller.

\subparagraph*{Construction.}
The tag array needs no suffix array.
One pass of LF over the text gives the suffix-array samples at the BWT run boundaries, as for the r-index, and from them $\phi$ streams the suffix array from the last BWT position to the first; each entry is mapped to its taxon with a predecessor search, and consecutive positions with the same taxon are merged into runs, which are written to disk as they are completed.
The memory needed is therefore proportional to $r + d$, plus $\rho$ for reading the runs back.
On SILVA this takes about 6 minutes and 3.3\,GB on the desktop, and on the homopolymer-compressed Zymo reference 7.2 hours and 8.9\,GB on the laptop (with a version of the construction that kept all the runs in memory).

\subparagraph*{Grammar-compressing the runs.}
\label{sec:grammar}
Our default configuration uses only run-length compression, but the sequence of the runs' tags can be compressed further.
LF maps the positions $i..j$ of a BWT run to consecutive positions $\mathrm{LF}(i)..\mathrm{LF}(j)$, whose suffixes are one character longer, so the tags at positions $\mathrm{LF}(i)..\mathrm{LF}(j)$ are a copy of those at positions $i..j$, except at the positions $\mathrm{LF}(q)$ such that $\mathrm{SA}[q]$ is the first position of a block.
The copying holds because the tag of a suffix does not change when the suffix is extended by one character within its block.
It fails for tag arrays in general: the BWT itself is a tag array, in which each suffix is tagged with the character preceding it, and it is run-length compressible, but LF does not map its runs to copies of themselves.

\begin{observation}
\label{obs:grammar}
Let $S[1..n]$, ending with a unique terminator, be cut into $d$ blocks, each with a tag, let $T[i]$ be the tag of the block containing $\mathrm{SA}[i]$, and let $r$ be the number of runs in the BWT of $S$.
Then $T$ has a bidirectional macro scheme with at most $r + 2d$ phrases, and therefore a string attractor of size at most $r + 2d$ and a grammar of size $O((r + d) \log (n / (r + d)))$.
\end{observation}

\begin{proof}
Unless $\mathrm{SA}[q]$ is the first position of a block, $T[\mathrm{LF}(q)] = T[q]$, since $\mathrm{SA}[\mathrm{LF}(q)] = \mathrm{SA}[q] - 1$ is in the same block (taking positions cyclically, so that $\mathrm{SA}[\mathrm{LF}(q)] = n$ when $\mathrm{SA}[q] = 1$).
Let $T[p]$ copy $T[\Psi(p)]$ for every such position $p = \mathrm{LF}(q)$, and store $T[p]$ explicitly at the other $d$ positions, at which $\mathrm{SA}[p]$ is the last position of a block.
LF maps the positions of each BWT run to consecutive positions, in order, so cutting $[1..n]$ at the ends of the images of the $r$ runs and around the $d$ explicit positions gives at most $r + 2d$ phrases, each a copy of a substring of $T$ or a single explicit symbol.
Following copies from any position $p$ moves to $\Psi(p)$, whose suffix starts one position later in $S$, so it reaches an explicit position before leaving the block containing $\mathrm{SA}[p]$, and the scheme is a valid bidirectional macro scheme.
The size $\gamma$ of the smallest string attractor of a string is at most the number of phrases of any bidirectional macro scheme for it, and the string has a grammar of size $O(\gamma \log (n/\gamma))$~\cite{KP18}; since $x \log (n/x)$ is increasing for $x \leq n/e$ and $\gamma \leq r + 2d$, the bound follows.
\end{proof}

We say that a property constant on few blocks of $S$, such as documents, genomes and taxa, has \emph{textual locality}.
Textual locality gives the macro scheme of Observation~\ref{obs:grammar}, with $O(r + d)$ phrases, while contextual locality gives few runs, by Lemma~\ref{lem:runs}.
The characters themselves have contextual locality but not textual locality, so the BWT is run-length compressible, but LF does not map its runs to copies of themselves.
Taxa have both, so their tag array has few runs and also a macro scheme with at most $r + 2d$ phrases, which is smaller when $\rho \gg r$; genomes, with many per species, have mainly textual locality.
Cobas and Navarro~\cite{CN19} used the copies in the proof of Observation~\ref{obs:grammar} to grammar-compress the document array, as with the differential suffix array~\cite{GN07,GNP20}.
We instead run-length compress first, and grammar-compress only the sequence of the runs' taxa (their lengths come from the run starts), which has length $\rho$ rather than $n$; Observation~\ref{obs:grammar} does not directly bound this grammar.
On SILVA we build the grammar with RePair~\cite{LM00}, in the second author's implementation distributed with BigRePair~\cite{GIM+19}, which avoids deep rules, and encode it as Gagie et al.\ do~\cite{GIM+20}; to count, we find the first run of the interval by a predecessor search on the final sequence, descend through the grammar, and decode the following runs.
The runs' taxa then take 0.23\,GB instead of 0.63\,GB, and the index 1.04\,GB instead of 1.44\,GB, while classification takes 75.0\,$\mu$s per read instead of 66.0, with identical answers.

\section{Time and space}
\label{sec:tradeoffs}
\label{sec:experiments}

\subparagraph*{SILVA.}
Our default index for SILVA takes 1.44\,GB: 0.57\,GB for the RLCSA, 0.86\,GB for the tag array (19.1 bits per run) and 18\,MB for the lookup table.
Building it from the sequences takes minutes on the desktop: 347\,s for the BWT, with prefix-free parsing~\cite{BGK+19} (peak memory 5.1\,GB), a few seconds for the RLCSA and the lookup table, and about 6 minutes for the tag array.
With one thread it classifies a read in 66\,$\mu$s, excluding loading the index and reading the reads, and with 12 threads on 6 cores, sharing one copy of the index, 144\,800 reads per second.
Table~\ref{tab:others} compares our index with Kraken~2 and Tagger, run on the same machine and reads, and with Cliffy's LCA query, run on our reads on the server.
There our accuracy for the default configuration is on fewer reads and its time is measured differently, so these numbers differ from those above (Appendix~\ref{app:details}).
On the server, our default index takes 1.42\,GB and 146\,$\mu$s per read and reaches 93.70\%, while Cliffy's undigested index takes 19.58\,GB and its LCA query takes 637\,$\mu$s and reaches 85.75\% (86.45\% with Cliffy's scoring); building it took almost three days, over 2\,TB of temporary disk space and more than 200\,GB of memory.
On Cliffy's digested index, which takes 7.63\,GB, it reaches 84.15\% (85.53\%), against 92.08\% on the server for our digested index, which takes 0.49\,GB.
With Cliffy's scoring, both are up to 4 points below what Cliffy's paper reports for its LCA query on its own reads (87.5 to 90.2\% and 86.3 to 88.5\%), as are the rz-index's results with Cliffy's matches on our reads (Appendix~\ref{sec:rz}), so our reads seem slightly harder; we have not run Cliffy's approximate listing, for which its paper reports 93.3 to 94.8\% and 92.8 to 93.4\%.
Appendix~\ref{app:listing} gives the time--space--accuracy Pareto surface, with indexes down to 0.41\,GB, or 0.16\,GB over minimizer digests.

\begin{table}[t]
\centering
\caption{SILVA: index size, time per read with one thread and reads per second with 12 threads (both excluding loading), and genus-level accuracy (mean over the four regions).
The first, fourth and fifth rows are on the desktop, with accuracy on 5\,000 read pairs per region and time as total time over the number of reads (200\,000 pairs per region); our other rows use the mean of per-read timings (66.0\,$\mu$s for the default) and 200\,000 pairs.
Tagger's three accuracies count pairs with disagreeing mates as wrong, random or correct.
Cliffy's numbers are on our reads (5\,000 pairs per region) and on the server, where our default configuration takes 146\,$\mu$s per read with the same accuracy as here.}
\label{tab:others}
\resizebox{\textwidth}{!}{\begin{tabular}{lrrrc}
\toprule
classifier & space (GB) & time ($\mu$s) & reads/s, 12 threads & accuracy \\
\midrule
long MEMs, proportional credit & 1.44 & 53.1 & 144\,800 & 93.70\% \\
\quad grammar-compressed runs & 1.04 & 75.0 & & \\
\quad minimizer digests & 0.49 & 39.3 & & 92.48\% \\
Kraken~2 & 0.15 & 14.3 & 366\,900 & 79.26\% \\
Tagger & 0.82 & 70.4 & 82\,600 & 81.69 / 87.24 / 92.78\% \\
Cliffy, undigested, LCA query & 19.58 & 636.6 & & 85.75\% \\
Cliffy, digested, LCA query & 7.63 & 419.1 & & 84.15\% \\
\bottomrule
\end{tabular}}
\end{table}

\subparagraph*{Zymo.}
For the homopolymer-compressed Zymo reference, the tag array takes 1.47\,GB, a run-length compressed BWT (RLBWT) 1.00\,GB and an explicit 64-bit RLCSA 5.56\,GB; with the RLCSA and the longest MEMs we classify 4\,900 whole reads per second with 12 threads on the laptop (peak memory 8.0\,GB).
The prefix test needs no tags, and with the RLBWT alone takes 109\,s for the 720-base prefixes, against the 194\,s and 0.74 to 1.54\,GB that SPUMONI~2 reports on its machine.

\section{Conclusions}
\label{sec:conclusions}

Classifiers based on exact matches differ in which matches they use and in how a match splits its credit among the taxa containing it.
On short 16S reads, all the MEMs of at least a given length with proportional credit are the most accurate choice; on long noisy reads from whole genomes, each read's longest MEMs, found by BML in one pass, avoid the threshold and are both more accurate and faster.
Proportional credit favors taxa with much sequence in the database, which ties its accuracy to how well the database matches the sample.
Long MEMs and their counts can be computed fast in compressed space, without locating any occurrence, with BML on an RLCSA and a complete, run-length compressed tag array, built from the BWT in space proportional to the number of runs.

The main open question is how to make credit robust to differences between the database and the sample, including reads from taxa absent from the database, for example by correcting towards an estimated composition of the sample; we would also like to scale to GTDB.
It would be natural to bound the grammar of the runs in terms of $r$, $d$ and $\rho$.
We conjecture that tagging each character with the rank of its genome in a phylogenetic tree gives a differentially compressible tag array that yields counts at every rank.
We also plan to redo MONI-align~\cite{VRF+25} with long MEMs and tags giving reference positions~\cite{BGG+24,OO26}.

\appendix
\section{Experimental details}
\label{app:details}

This appendix lists what we believe is needed to reproduce our results with our implementation, KATKA (\url{https://github.com/TravisGagie/KATKA}).
All parameters not mentioned are defaults.
On SILVA, the \emph{default configuration} finds MEMs of at least $L = 30$ bases with BML on an explicit RLCSA with a lookup table of 10-mers, and counts their occurrences with a run-length compressed tag array, for proportional credit.

\subparagraph*{Software and data.}
Cliffy's experiment scripts (\url{https://github.com/oma219/cliffy-experiments}, commit \texttt{566f647}), its read simulator MicrobeMixer (\url{https://github.com/oma219/MicrobeMixer}, \texttt{6409f5d}) with ART 2.5.8~\cite{HLMM12} and seqtk 1.4; Big-BWT (\url{https://github.com/alshai/Big-BWT}, \texttt{5b00650}); sdsl-lite (\texttt{c32874c}) with libdivsufsort.
Our own C++17 code was compiled with g++ 15.2 and \texttt{-O3 -march=native -DNDEBUG}, and the scoring scripts use Python 3.14.
The reference is SILVA 138.1 SSU Ref NR99 (\texttt{SILVA\_138.1\_SSURef\_NR99\_tax\_silva.fasta}, with \texttt{tax\_slv\_ssu\_138.1.txt} and \texttt{.tre}).
All times were measured with one thread on an Intel Core i7-8700 (3.2\,GHz, 6 cores) with 16\,GB of RAM under Linux.

\subparagraph*{Running the code.}
From the root of the repository, \texttt{demo\_quick.sh} fetches and builds the dependencies at the commits above, downloads SILVA, builds the reference and the default index (\texttt{silva/build\_katka.sh}), and classifies 5\,000 read pairs per region that are included in the repository, with one thread and with all cores.
\texttt{demo\_full.sh} builds every index in this paper (\texttt{ALL=1 silva/build\_katka.sh}), downloads the read pairs we used (every 50th pair of our simulation, 200\,000 per region, from a release of the repository) and runs every experiment, each with its own script in \texttt{silva/}, and \texttt{pareto\_ft.py} computes the Pareto surface.
To classify reads with the default configuration and 8 threads, the command is \texttt{RZ\_COUNTS=1 rz-classify -j 8 -L 30 -l -C bac.csa -T bac.s1.tag -F 10 bac.rz reads.fq out.listings}.
With a digested index, add \texttt{-B bac.map} and use \texttt{-F 2}.
The output has Cliffy's listing format, extended with the number of occurrences in each genus, so our \texttt{rz\_score2.py}, or Cliffy's scripts without the counts, can score it.
Finally, \texttt{silva/run\_cliffy.sh} and \texttt{silva/run\_katka.sh} run Cliffy and our implementation on the same reads on one machine, first with one thread each and then with the same number of threads each, check that the parallel outputs equal the single-thread ones, and score both.

\subparagraph*{Reference text.}
We converted the FASTA to one line per sequence (\texttt{seqtk seq -U}) and ran Cliffy's \texttt{write\_silva\_genera.py} with the tree and rank files and \texttt{-n 10000}, which writes one file per genus, numbered in the left-to-right order of the SILVA tree (9\,118 files), and \texttt{doc\_to\_traversal.txt}.
As in Cliffy's pipeline, each file is reverse-complemented with \texttt{seqtk seq -r}, which also turns U into T.
We call genus $i$ document $i$.
In the text $S$ of Section~\ref{sec:index}, $G_i$ is the concatenation of the sequences of genus $i$ separated by \texttt{X}, and any character other than A, C, G, T (in either case) becomes \texttt{X}.
A text position is mapped to its document by a predecessor search over the document starts.
For Cliffy's scoring script we also need \texttt{trav\_to\_length.txt}, which gives for each genus its traversal string and the total length of its sequences.

\subparagraph*{Reads.}
For each of V1--V2, V3--V4, V4 and V4--V5 we ran the \texttt{simulate} command of MicrobeMixer's \texttt{sim\_16s\_reads.py} with the region's primer file, the SILVA FASTA and taxonomy, 10 million reads, and the aquatic abundance profile shipped with MicrobeMixer (\texttt{data/aquatic.tsv}, in which we replaced \texttt{Candidatus\_} by \texttt{Candidatus} followed by a space, to match SILVA 138.1).
This gives 10 million pairs of 250-base reads and a truth file \texttt{*\_seqtax.txt} per region.
The read files are grouped by genus, so we subsample regularly: for accuracy, pairs $0, 50, 100, \ldots$ (200\,000 pairs per region); for time, mate-1 reads $0, 5000, 10000, \ldots$ (2\,000 reads per region), which are among the former.
Each mate is classified on its own, as given (both strands are in the index).

\subparagraph*{Lookup tables.}
For each string of length $t$ over the alphabet (A, C, G, T, or the digest bytes) the table stores the BWT interval after searching for it, and for the sr-index also the toehold.
It is built when the classifier starts (\texttt{-F} $t$), in under a second.
A search for a string of length at least $t$ starts with a lookup of its last $t$ symbols.
If they do not occur, the search proceeds symbol by symbol from scratch, so that BML learns exactly where it fails.
Every configuration gives the same answers with and without a table.

\subparagraph*{Digests.}
Number A, C, G, T as 0--3 and let a $k$-mer's code be its $2k$-bit integer.
Its hash is the splitmix64 finalizer applied to the smaller of its code and its reverse complement's code.
In each maximal A/C/G/T run, every window of $w$ consecutive bases ($w - k + 1$ consecutive $k$-mers) selects \emph{all} its $k$-mers of minimum hash.
The selected $k$-mers, in text order and in their own orientation, form the digest of the run, after collapsing equal consecutive symbols.
Digests of consecutive runs are joined by \texttt{X}.
With $k = 4$ and $w = 11$, every one of the 256 codes occurs, but only 252 bytes are free (0--2 are reserved for terminators and \texttt{X} is the separator).
The rarest codes in the reference, taken in reverse-complement pairs until the rest fit (6 codes), therefore share one byte, and the others get bytes $3, 4, \ldots$ (skipping \texttt{X}) in increasing order of code.
Reverse complementing a digest reverses it and maps each byte to the byte of the reverse-complemented $k$-mer.
We digest $G_i$ and build $S$ from the digests as above.

\subparagraph*{Indexes.}
The BWT of $S$ (with one terminator) is built with Big-BWT.
The RLBWT stores the run heads in a Huffman-shaped wavelet tree and the run starts in Elias--Fano coded bitvectors (one overall and one per character).
The RLCSA stores, for each character $c$, the BWT positions of the starts of the runs of $c$ and the cumulative lengths of those runs, either as plain 32-bit arrays or with Elias--Fano coding.
A backward step by $c$ is a predecessor search in the list of $c$, and gives the same intervals as the RLBWT.
The tag array is built without a suffix array: one pass of LF over the text gives the suffix-array samples at the BWT run boundaries, as for the r-index, and from them $\phi$ streams the suffix array from the last BWT position to the first.
Each entry is mapped to its document by a predecessor search over the document starts, and consecutive positions with the same document are merged into runs (the terminator's position gets the last document), so the memory needed is proportional to the number of BWT runs plus the number of tag runs.
The r-index~\cite{GNP20} stores the suffix-array samples at the run boundaries and supports $\phi$ with a predecessor structure on them.
Our sr-index removes samples as in Cobas et al.~\cite{CGN21,CGN25} with parameter $s$, and $s = 0$ means no subsampling.
Reported sizes are file sizes in units of $10^9$ bytes.
A configuration's space is the sum of the structures it uses: the RLBWT or RLCSA used for backward search; the tag array (plus the RLBWT, for LF, if it is sampled and the search uses an RLCSA); the sr-index samples (plus the RLBWT, for reading run heads, if the search uses an RLCSA); and the lookup table.

\subparagraph*{BML without digests.}
Let $R[0..m - 1]$ be a maximal A/C/G/T run of the read, skipped if $m < L$; set $x = 0$ and, while $x + L \le m$:
search backward for $R[x..x + L - 1]$ from its right end; if the step for $R[j - 1]$ fails, set $x = j$ and repeat.
Otherwise continue backward from $x$ to the leftmost $a$ such that $R[a..x + L - 1]$ occurs, then search forward from $a$ (backward on the reverse complement) to the largest $e$ such that $R[a..e - 1]$ occurs.
Report the MEM $R[a..e - 1]$ (the genera of its occurrences are those of its reverse complement's, whose interval the forward search gives), and set $x = e - L + 1$.

\subparagraph*{BML with digests.}
For a maximal A/C/G/T run $R[0..m - 1]$ of the read with $m \ge \max(L, w)$, digest it as above, and for each window $y$ (covering bases $y..y + w - 1$) let $\mathit{lo}[y]$ and $\mathit{hi}[y]$ be the positions in the digest $\hat{R}$ of the first and last symbols it selects; both are non-decreasing in $y$.
Set $x = 0$ and, while $x + L \le m$: search backward for $\hat{R}[\mathit{lo}[x]..\mathit{hi}[x + L - w]]$; if the step for $\hat{R}[j - 1]$ fails, set $x$ to the smallest $y$ with $\mathit{lo}[y] \ge j$ (stop if there is none) and repeat.
Otherwise extend backward to the leftmost $a'$ and then forward to the largest $e'$ such that $\hat{R}[a'..e' - 1]$ occurs.
The MEM's span in the read is the interval of bases $[\mathit{left}, \mathit{right}]$, where $\mathit{left}$ is the smallest $y$ with $\mathit{lo}[y] \ge a'$, and $\mathit{right} = y_R + w - 2$ for the smallest $y_R$ with $\mathit{hi}[y_R] \ge e'$ ($\mathit{right} = m - 1$ if there is none).
Record the MEM with its span, and set $x = \max(x + 1, y_R - (L - w))$.
Finally, discard every recorded MEM whose span is contained in another's, and report the rest with their spans.

\subparagraph*{Answers and scoring.}
The listing answer for a MEM $M$ is the set $D$ of documents containing it, with the number $c_g$ of occurrences in each ($\texttt{RZ\_COUNTS=1}$).
With proportional credit, each $g \in D$ gets weight $|M| c_g / \sum_{h \in D} c_h$, and with even credit $|M|/|D|$.
With the range scheme the answer is the pair $\{g_\ell, g_r\}$ of the smallest and largest documents containing $M$, and each document from $g_\ell$ to $g_r$ gets weight $|M|/(g_r - g_\ell + 1)$.
This is Cliffy's LCA scoring (``method0'' in its scripts), since documents are in tree order.
Here $|M|$ is the length of $M$ in bases (of its span, with digests).
Weights from both mates are summed per genus (through \texttt{doc\_to\_traversal.txt}), with Cliffy's \texttt{DocProfReadSet} class handling the truth files.
A pair whose true genus is missing or ambiguous in the truth file counts as correct, as in Cliffy's script, under both of our measures.
Otherwise a pair is correct under Cliffy's rule if its true genus has the maximum weight, possibly tied, or if it has no weight at all and its true genus is not ``uncultured''.
It is \emph{strictly} correct if its true genus alone has the maximum weight.
We checked that our re-implementation of the scoring gives the same numbers as Cliffy's \texttt{classify\_silva\_readset.py} on a sample of 1\,500 pairs.

\subparagraph*{Cliffy.}
The second author built Cliffy's undigested and digested indexes with its authors' settings, and ran them and our default configuration on our 5\,000 read pairs per region, with \texttt{silva/run\_cliffy.sh} and \texttt{silva/run\_katka.sh}, on a server with two Intel Xeon E5-2407 v2 processors (8 cores in all, 2.4\,GHz) and 256\,GB of RAM.
Cliffy's two-pass construction of the undigested index writes about 2.3\,TB of temporary files (Cliffy's Table~1), more than the server's home disk holds: a first attempt failed after 40 hours, when the disk filled, having used 240\,GB of memory, and the build succeeded in 69 hours with the temporary files on a larger disk.
Classifying with it took 37\,GB of memory.
Times per read are with one thread or process, excluding loading the index (50\,s for Cliffy's undigested index), and accuracies are for Cliffy's LCA query (``method0''), scored as for our implementation.

\subparagraph*{Other classifiers.}
We ran Kraken~2~\cite{WLL19} and Tagger~\cite{DAF+25} on the same reads and the same machine as our implementation (\texttt{silva/run\_kraken2.sh} and \texttt{silva/run\_tagger.sh}).
We built Kraken~2's database with its own SILVA script (\texttt{kraken2-build --special silva}, default $k = 35$ and minimizer length 31), as Cliffy's authors did.
It downloads the same release, SILVA 138.1 SSU NR99, but indexes all of its sequences, including those not annotated down to the genus level, with SILVA's taxonomy.
Kraken~2 classifies the two mates of a pair together (\texttt{--paired}), and we scored its classifications with the part of Cliffy's script for Kraken~2, under which a pair is correct at the genus level only if Kraken~2 assigns it to its true genus.
Pairs assigned to a higher taxon, to the root or to nothing count as wrong (Cliffy's rule for pairs with no votes applies only to its own methods).
Counting also pairs assigned to an ancestor of their true genus below the domain, which that script calls vague positives, Kraken~2 reaches 84.3\%.
Tagger's authors instead built Kraken~2 with a flat hierarchy, in which a read is assigned to a genus, to the root or to nothing, so our numbers for Kraken~2 are not comparable with theirs.
For Tagger, we built an index of the same 402\,378 sequences as ours, with each sequence's genus as its tag and one strand only (Tagger handles reverse complements when classifying), and used its default parameters ($L = 25$, a 10-mer lookup table).
With $L = 30$, as for our default configuration, Tagger was no more accurate (81.05\%, 86.88\% and 92.72\% under the three rules), so we report its defaults.
Tagger assigns each read a single genus, so we classified the two mates separately and let each give its length to its genus, as we do for MEMs.
A pair whose mates are assigned different genera is then a tie, which is wrong under the strict measure, and a pair with one mate unclassified gets the other's genus.
(Tagger's authors classified single reads.)
For both tools, the times are with one thread and exclude loading the index: Kraken~2 reports its own classification time, and for Tagger we subtracted the running time of an invocation on a single read, which is dominated by loading the index.
Tagger breaks ties between genera at random, with a separate random number generator in each thread, so its outputs with one thread and with 12 differ slightly, and we scored the single-thread outputs.
If a pair's mates are assigned different genera, breaking the tie at random gives the true genus with probability 1/2 when one of them is right, so the expected accuracy is the mean of the accuracies with such pairs counted as wrong and as correct.
We measured speed separately, with \texttt{silva/run\_scale.sh}, on 200\,000 read pairs per region with 1, 6 and 12 threads, because on 5\,000 pairs a run with 12 threads takes well under a second and start-up dominates: on 5\,000 pairs, Kraken~2 was no faster with 12 threads than with one.
With 6 and 12 threads, our implementation was 5.3 and 7.7 times faster than with one, Kraken~2 4.3 and 5.2 times, and Tagger 4.5 and 5.8 times.
Our time per read in Table~\ref{tab:others}, 53.1\,$\mu$s, is its total classification time divided by the number of reads.
The 66.0\,$\mu$s elsewhere in the paper, which we use to compare our configurations, is measured differently, as the mean of per-read timings on 2\,000 reads per region, and we have not investigated the difference.
On the 5\,000 read pairs per region, the default configuration has strict accuracy 93.70\%, against 93.77\% on the 200\,000 pairs per region of Section~\ref{sec:accuracy}.

\subparagraph*{Credit schemes.}
The SILVA column of Table~\ref{tab:credit} scores one set of listings, made with \nolinkurl{RZ_COUNTS=1} and $L = 30$ for every 50th read pair, under each scheme with \nolinkurl{silva/rz_score_schemes.py}: \texttt{freq} is proportional credit, \texttt{list} is even credit, \texttt{norm} divides each count by twice the number of sequences of the genus, \texttt{one} gives the whole credit to one genus drawn with probability proportional to its count, \texttt{cliffs} gives each genus containing the MEM a uniformly random value and splits the credit evenly over the left-to-right and right-to-left maxima of these values in genus order, and \texttt{lca} spreads it over the range.
Random choices use a fixed seed.
The range row is from the rz-index run of Appendix~\ref{sec:rz}, which gives the same answers as \texttt{lca}.
The Zymo column scores the whole-read listings in the same way with \nolinkurl{tools/credit_schemes.py}, at the species level and per read, and with \nolinkurl{tools/norm_credit.py}, which divides each count by the number of reference sequences of its species, or by their total length.
In every case a read or read pair is correct only if its true taxon alone has the largest weight.

\subparagraph*{Zymo.}
We used the files that N.~Brown deposited on Zenodo (record 23251764): \nolinkurl{microbial_ref.fa.gz}, 14\,280 RefSeq sequences of seven bacterial species; \nolinkurl{microbial_reads.fa.gz}, 581\,802 ONT reads with their species from minimap2 alignments (\nolinkurl{microbial_reads.fa.class} and \nolinkurl{class.txt}); and \nolinkurl{yeast_reads.fa.gz}, 10\,245 yeast reads from the same run.
We tagged each sequence with its species, wrote one file per species, and split it into files of at most 1\,GB, keeping the sequences in their original order, so that every file is a dataset of a single species.
\texttt{rz-prep -d -s} (with \texttt{-U} for homopolymer compression) builds the text with both strands; Big-BWT builds the BWT of each dataset, single-threaded, and pfp-merge~\cite{BGK+19} merges them; \texttt{rz-build} builds the r-index samples, \texttt{rz-tagbuild -S} the species tag array, and \texttt{rz-csabuild} the explicit RLCSA.
With homopolymer compression, $S$ has 39.28 billion characters, 73\% of the 53.74 billion without it, the BWT has 338.1 million runs (436.6 million without compression) and the tag array 448.7 million runs.
On the laptop, the largest Big-BWT run (\emph{E.~coli}, 18.2\,Gbp after compression) took 982\,s and 7.75\,GB, the merge 1\,240\,s and 7.3\,GB, the r-index samples 6\,436\,s and 6.0\,GB, and the tag array 26\,052\,s and 8.9\,GB, with the version of the construction that keeps all the runs in memory.
We classify with \nolinkurl{RZ_COUNTS=1 rz-classify -U -L 30 -D 14 -j 12}, which reads FASTA or FASTQ; $L$ and the filter's threshold apply to the compressed read, and reported spans are mapped back to the original read.
Without homopolymer compression we omit \nolinkurl{-U} in both programs and use \nolinkurl{-D 16}.
For the longest MEMs we add \nolinkurl{-X} with \nolinkurl{-L 1}, or with \nolinkurl{-L 20 -K zy.k20.kbf} for the KeBaB filter, built with \nolinkurl{rz-kbbuild} ($k = 20$, target false-positive rate 0.1) in 2\,736\,s and 0.15\,GB; \nolinkurl{tools/maxmem_roc.py} scores each read by the compressed length of its longest MEM, computed from the read, and prints the sensitivity, specificity and species accuracy for each threshold.
For SPUMONI~2's Table~2 we took the first 180, 360, 540 and 720 bases of each bacterial and yeast read: a prefix is accepted if it has at least one MEM of length at least $L$ that passes the filter, so the sensitivity is the fraction of bacterial prefixes accepted and the specificity the fraction of yeast prefixes rejected.
The species accuracy of the prefixes is with proportional credit.

\subparagraph*{Timing.}
Times are wall-clock per read, measured inside the classifier, with the index already loaded and the reads already read from disk.
They include digesting the read, finding the matches and answering them.
We report the mean over the four regions of the mean over the 2\,000 reads, all measured in a single run of \texttt{run\_pareto\_ft.sh} (with \texttt{COUNTS=1}), which also checks that every configuration gives the same answers with and without a lookup table.
For every configuration, we checked that all backends and sampling rates give identical outputs.

\subparagraph*{Parallelism and further systems optimizations.}
With \texttt{rz-classify -j}~$p$, $p$ threads share one copy of the index: each takes chunks of 256 reads from the input, and the answers are written back in input order, so the output is identical to that of a single thread (we checked this for every listing method, with and without lookup tables and digests).
The index structures are read-only during classification, and the scratch space for document listing, counting and locating, and the statistics counters, are kept per thread.
Alternatively, \texttt{silva/run\_split.sh}~$p$ splits the read file into $p$ contiguous slices, classifies them with $p$ processes and concatenates the outputs.
This needs no threads but loads $p$ copies of the index.
Beyond parallelism, we have not tuned our code at the systems level.
Backward search and LF-mapping are dominated by cache misses that depend on each other, so the natural next steps are interleaving the searches for several reads and prefetching their next memory accesses, to hide the latency of each miss behind work on other reads; huge pages, to reduce TLB misses on indexes of a few gigabytes; and memory-mapping the indexes, so that processes share them and start immediately.
Graphics processors suit data-parallel work better than pointer chasing, but they can also hide latency by running very many searches at once, so they might pay off for large batches of reads.

\section{Minimizer digests}
\label{sec:digest}

Like Cliffy, our implementation can index minimizer digests~\cite{RHH+04,ARG+23} of the sequences instead of the sequences themselves.
We digest each sequence and its reverse complement as Cliffy does, with $k = 4$ and $w = 11$, but symmetrically, so that the digest of a reverse complement is the reverse complement of the digest: we use a canonical hash, select all minimizers tied for the minimum in a window, and collapse equal consecutive minimizers.
Since 4 byte values are reserved for terminators and the separator, the 256 distinct minimizers fit in a byte only once the 6 rarest share one.
The digested text has 262 million symbols and its BWT has 35.9 million runs.
To classify a read, we digest it the same way, keeping only the minimizers selected by windows lying entirely inside each match.
A window inside a match in the read is also a window inside the corresponding occurrence in the database, and selects the same minimizers there, so the digest of every DNA match occurs in the digested database.

We keep the threshold $L$ in bases rather than in digest symbols, since consecutive digest symbols can stand for very different numbers of bases.
A window of the read covers $w$ bases and selects a range of digest symbols, and these ranges move monotonically to the right as the window does.
A DNA match $R[x..x + L - 1]$ contains exactly the windows starting in $[x, x + L - w]$, so its digest is the range of symbols they select.
BML searches backward for this range of symbols.
If the search fails at symbol $j$, then every window start $x' \geq x$ whose first window selects a symbol no later than $j$ also fails, and we jump to the first window that selects only symbols after $j$.
If the search succeeds, we extend the match to a maximal digested MEM and compute its span in bases: the longest DNA interval of the read all of whose windows select symbols within the MEM.
This finds every digested MEM whose span is at least $L$ bases and is not contained in another digested MEM's span.
We checked this against brute-force enumeration on 3\,000 reads for every $L$ from 15 to 50.
The weight of a MEM is the length of its span.

\begin{table}[t]
\centering
\caption{Classification over minimizer digests ($k = 4$, $w = 11$), with tags and proportional credit, against the default configuration and Cliffy's digested index.
All use a lookup table ($t = 2$ digest symbols with digestion, $t = 10$ bases without).}
\label{tab:digest}
\small
\begin{tabular}{lrrrrrr}
\toprule
 & space & \multicolumn{5}{c}{time per read ($\mu$s)} \\
\cmidrule{3-7}
configuration & (GB) & $L = 15$ & $L = 20$ & $L = 30$ & $L = 40$ & $L = 50$ \\
\midrule
\multicolumn{2}{l}{\emph{digests}: strict accuracy} & 92.47\% & 92.47\% & 92.48\% & 92.36\% & 92.05\% \\
\quad explicit RLCSA, tags & 0.49 & 56.6 & 54.7 & 39.3 & 32.7 & 25.2 \\
\quad Elias--Fano RLCSA, tags & 0.27 & 79.0 & 76.6 & 61.9 & 50.7 & 44.4 \\
\quad RLBWT, tags sampled with $s = 4$ & 0.19 & & & 220.1 & 147.0 & 114.3 \\
\addlinespace
\multicolumn{2}{l}{\emph{undigested}: strict accuracy} & 93.74\% & 93.77\% & 93.77\% & 93.48\% & \\
\quad default configuration & 1.44 & 88.6 & 79.2 & 66.0 & 57.3 & \\
\midrule
\multicolumn{7}{l}{\emph{Cliffy}~\cite{ABL25}, digests, approximate listing: 92.8--93.4\% with their scoring and reads} \\
\quad & 4.27 & \multicolumn{5}{c}{77--82} \\
\bottomrule
\end{tabular}
\end{table}

Table~\ref{tab:digest} shows that digestion makes the index three times smaller (0.49 instead of 1.44\,GB with an explicit RLCSA and tags) and classification 1.7 times faster with $L = 30$ (39.3 instead of 66.0\,$\mu$s), at a cost of 1.3 points of accuracy (92.48 instead of 93.77\%).
The digested tag array has 81.0 million runs, 2.26 times as many as the BWT, and takes 0.19\,GB (18.3 bits per run), plus 0.03\,GB for the range-minimum structure if listing.
Over digests, an Elias--Fano RLCSA is only 1.6 times slower than an explicit one, so it gives a 0.27\,GB index that is still faster than the undigested default, and sampled tag arrays on the RLBWT give the smallest indexes we built, 0.16 to 0.19\,GB.
Once the sequences are digested, building the index with an explicit RLCSA takes about 4 minutes: 132\,s for the BWT (peak memory 2.0\,GB), a few seconds for the RLCSA, and 79\,s for the tags.
Cliffy's paper reports a digested index of 4.27\,GB, almost 9 times larger, about twice as slow per read on its server, with 92.8 to 93.4\% accuracy for approximate listing with its own scoring on its own reads.
On our reads, Cliffy's LCA query on its digested index reaches 84.15\% strict accuracy (Table~\ref{tab:others}), against 92.08\% for our digested index on the same reads and server (92.48\% on 200\,000 read pairs per region).

\section{Alternatives we considered}
\label{app:alternatives}
\label{app:negative}

This appendix explains the choices behind our default configuration, by describing the alternatives we tried along the way, and some we have not tried yet.
Most of them were worse on SILVA, which consists of short, highly conserved sequences with many per genus.
Some might work better on other datasets, such as whole genomes, where the number of occurrences of matches and the compressibility of tag arrays could differ greatly.

\subsection{Which matches}
\label{app:matches}

Cliffy cuts each read into \emph{Ziv--Merhav phrases} (ZMPs), the phrases of the cross parse of Section~\ref{sec:background}: each is the longest string ending at its position that occurs in the database, and each new backward search starts at the character where the previous one failed.
The same parse underlies Relative Lempel--Ziv compression~\cite{KPZ10}, but there each phrase also stores the position of a source, which is needed to decompress.
For classification only the phrases matter, so ZMP is the more appropriate name.
Our first experiments compared ZMPs with all the MEMs of each read, found by Li's forward--backward algorithm~\cite{Li12}, which alternates backward and forward extensions and starts each search one position after the end of the previous MEM, with the range scheme.
MEMs were more accurate than ZMPs by 0.4 to 1.3 points at the genus level (with Cliffy's scoring), both with and without digestion.
However, reads have many short MEMs, short MEMs occur many times in the database, and listing the genera containing them is expensive: with a subsampled r-index (Appendix~\ref{app:listing}), locating all the occurrences of all the MEMs of a read took 35 to 78\,ms.

Table~\ref{tab:L} shows how the threshold $L$ of BML affects accuracy with even credit.
With $L = 15$, the accuracy of the range scheme (87.3\% with Cliffy's scoring; not shown in the table, which is for even credit) is the same as with all the MEMs found by forward--backward search (87.2\%), while listing with the sr-index becomes more than 100 times faster (289\,$\mu$s per read instead of 35 to 78\,ms, with sampling parameter $s = 4$ and an explicit RLCSA).
Accuracy is almost unchanged up to $L = 30$, drops by about 2 points at $L = 50$ and collapses for larger $L$, as reads start to have no MEM that long.
Most of the gain in time comes from not listing the short MEMs, but BML also searches faster: forward--backward search that lists only the MEMs of length at least $L$ takes 133 to 146\,$\mu$s per read without digestion, almost independently of $L$, against 77 to 136\,$\mu$s for BML (unsampled tags, explicit RLCSA, no lookup table, $L$ from 40 down to 15), and remains slower with a lookup table (for example, 75 against 61\,$\mu$s with $L = 30$ and a 12-base table).

\begin{table}[t]
\centering
\caption{Genus-level accuracy with BML for various thresholds $L$ (mean over the four regions), listing the genera containing each MEM, with even credit.
``Strict'' counts ties and unclassified read pairs as wrong; ``theirs'' is Cliffy's scoring.
The last column is the fraction of unclassified read pairs.}
\label{tab:L}
\begin{tabular}{rccccccc}
\toprule
 & \multicolumn{3}{c}{undigested} & & \multicolumn{3}{c}{minimizer digests} \\
\cmidrule{2-4} \cmidrule{6-8}
$L$ & strict & theirs & uncl. & & strict & theirs & uncl. \\
\midrule
15 & 90.56\% & 91.59\% & 0.0\% & & 88.93\% & 89.97\% & 0.0\% \\
20 & 90.53\% & 91.68\% & 0.0\% & & 88.91\% & 89.98\% & 0.0\% \\
30 & 90.28\% & 91.83\% & 0.0\% & & 88.84\% & 90.03\% & 0.0\% \\
40 & 89.61\% & 91.93\% & 0.0\% & & 88.55\% & 90.12\% & 0.0\% \\
50 & 88.27\% & 92.01\% & 0.2\% & & 87.91\% & 90.27\% & 0.0\% \\
75 & 78.05\% & 87.46\% & 7.6\% & & 82.59\% & 89.91\% & 2.3\% \\
100 & 56.85\% & 66.13\% & 31.8\% & & 67.58\% & 77.77\% & 18.0\% \\
\bottomrule
\end{tabular}
\end{table}

The table also shows why we prefer the strict accuracy: with Cliffy's scoring, accuracy appears to \emph{increase} up to $L = 50$, but only because more read pairs end in ties among a few genera.
Even credit is also clearly more accurate than the range: with $L = 15$, by 3.5 points (strict) without digestion and by 3.6 points with it.

\subsection{How to give credit}
\label{app:freq}
\label{sec:norm}

\begin{table}[t]
\centering
\caption{Genus-level accuracy (strict / Cliffy's scoring, mean over the four regions) with even credit to each genus containing a match, credit proportional to the number of occurrences in each genus, and credit proportional to the fraction of each genus's sequences (both strands) containing the match.}
\label{tab:freq}
\begin{tabular}{lccc}
\toprule
 & even & proportional & normalized \\
\midrule
digests, $L = 15$ & 88.93 / 89.97\% & 92.47 / 92.49\% & 91.36 / 91.36\% \\
digests, $L = 30$ & 88.84 / 90.03\% & 92.48 / 92.51\% & 91.34 / 91.34\% \\
undigested, $L = 15$ & 90.56 / 91.59\% & 93.74 / 93.75\% & 92.92 / 92.92\% \\
undigested, $L = 30$ & 90.28 / 91.83\% & 93.77 / 93.79\% & 92.76 / 92.76\% \\
\bottomrule
\end{tabular}
\end{table}

Even credit, as in Cliffy's approximate listing, was our first choice, since it treats each genus containing a match alike.
Table~\ref{tab:freq} shows that proportional credit (Section~\ref{sec:accuracy}) is about 3.5 points more accurate, both with and without digestion, and that it almost eliminates ties.
Proportional credit favors genera with many sequences in SILVA, since they have more occurrences of everything.
To separate this effect from that of weighting each genus by how much of it contains the match, we also tried dividing each count $c_g$ by $2 n_g$, where $n_g$ is the number of sequences in genus $g$ and the factor 2 accounts for the two strands, so that each genus gets credit in proportion to the fraction of its sequences containing the match.
This normalized credit gains about 2.5 points of strict accuracy over even credit, but about 1 point less than proportional credit.

This does not mean that normalizing is a bad idea in general.
MicrobeMixer chooses the genus of each simulated read according to the abundances in an aquatic biome, and only then a sequence of that genus, so our reads are not drawn in proportion to how well each genus is represented in SILVA.
Proportional credit then acts as a prior that a genus with many sequences in the database is more likely, and it helps only insofar as the genera that are abundant in the simulated samples are also well represented in SILVA, which is plausible.
On samples whose composition differs from that of the database, the prior could hurt, and normalizing, or correcting towards a known or estimated composition of the sample, could be better.
A direct test would be to simulate reads by choosing genera uniformly, regardless of their representation in the database, and then sequences within each genus: if the advantage of proportional credit comes from the prior, it should shrink or disappear, while both kinds of weighting should still beat even credit if the rest of the gain comes from weighting each genus by how much of it contains the match.
We leave this test for future work.
On the Zymo reads, where the reference has up to 30 times as many sequences of one species as of another, normalizing by the number of sequences or by their total length did not help (Section~\ref{sec:accuracy}), but almost every long MEM there occurs in a single species.

\subsection{How to list and count}
\label{app:listing}
\label{sec:sr}
\label{sec:rlcsa}

\subparagraph*{Subsampled r-indexes.}
Once we have the BWT interval of a MEM, the obvious way to find the genera containing it, and the number of occurrences in each, is to locate all its occurrences with an r-index~\cite{GNP20} and map each to its genus.
The r-index stores two samples of the suffix array per run of the BWT.
The subsampled r-index (sr-index)~\cite{CGN21,CGN25} discards samples that are close in the text to others, so that it stores $O(\min(r, n/s))$ samples for a parameter $s$ while locating each occurrence in $O(s)$ time.
We first stored the BWT as an RLBWT, with the run heads in a Huffman-shaped wavelet tree and the run starts in Elias--Fano coded bitvectors.
With $s = 16$ and no lookup table, this index takes only 0.41\,GB and classifies a read in 477\,$\mu$s with $L = 30$.
Building it from the sequences took under 9 minutes on our desktop (347\,s for the BWT, 119\,s for the r-index samples and 43\,s to subsample them), and it is the smallest undigested index on our Pareto surface.

\subparagraph*{RLCSAs.}
Backward search on an RLBWT computes rank queries on the BWT, which become slower as the alphabet grows, and this matters especially for minimizer digests, whose alphabet has 256 symbols.
The RLCSA (Section~\ref{sec:index}) avoids this, and the sr-index's subsampling works the same way on an RLCSA, as in Cobas et al.'s sr-csa~\cite{CGN25}.
We still use the RLBWT's run heads to read characters when locating.
Building an RLCSA from the BWT takes a few seconds, and an explicit RLCSA makes every configuration two to four times faster than the RLBWT.

\subparagraph*{Tag arrays.}
A tag array needs no suffix-array samples, and listing or counting with it costs time proportional to the number of genera or runs in the interval, rather than to the number of occurrences, so it is much faster than locating.
It is also larger, because the tags have more runs than the BWT.
Tagger~\cite{DAF+25} stores the genus of only the last position of each BWT run, which suffices for its single-genus toehold but not for listing.
Most of the space of our tag array is in the genera of its runs, so we tried storing them for only some runs, in the spirit of subsampled r-indexes~\cite{CGN25} and of Olbrich and Ohlebusch's sampled tag arrays~\cite{OO26}.
Following Olbrich and Ohlebusch, let $e(i)$ be the run containing the LF-mapping of the first position of run $i$.
LF moves to the previous position of the text, so run $e(i)$ has the same genus as run $i$, unless the first position of run $i$ is the first position of a block.
We store the genera of those runs, of which there are fewer than genera, since only some genus starts are the first position of a run.
The graph with edges $i \to e(i)$ has out-degree 1.
Unlike in Olbrich and Ohlebusch's setting, where tags strictly decrease along edges, it may have cycles (within a genus), so we also store the genus of one run on each cycle.
We then choose the other sampled runs greedily so that every run reaches a sampled run in fewer than $s$ steps.
On the digest there are 7\,467 block boundaries and 11\,141 cycles, and we sample 24\% of the runs with $s = 4$ and 5\% with $s = 16$, which shrinks the tags from 0.19\,GB to 0.09\,GB and 0.06\,GB, respectively (without the range-minimum structure).
Sampling did not help much, and grammar-compressing the runs instead works better (Section~\ref{sec:tags}).
Computing LF needs access to the BWT, which an RLCSA does not provide, so a sampled tag array needs the RLBWT as well, which cancels most of its savings unless the RLBWT is used for backward search anyway.
Moreover, counting with sampled tags needs the genus of every run in the interval, each recovered by following LF.

\subparagraph*{The Pareto surface.}
A configuration is on the time--space--accuracy Pareto surface if no other configuration is at least as good in all three measures and better in at least one.
Table~\ref{tab:pareto} lists the configurations on this surface, out of all combinations of backend (RLBWT, Elias--Fano or explicit RLCSA), listing structure (sr-index with $s = 0, 4, 16$, tags with $s = 1, 4, 16$, or, without digestion and with $L = 30$ or 40, tags with grammar-compressed runs), lookup table ($t = 0, 10, 12$ bases or $0, 2$ digest symbols) and $L$ (15 to 40 bases, or to 50 with digests), all with proportional credit, so the space of a tag array does not include its range-minimum structure.
No configuration with $L < 30$ is on the surface, since $L = 30$ is at least as accurate and faster.
Without digestion, sr-indexes give the smallest configurations, run-length compressed tag arrays the fastest, and grammar-compressed runs most of those in between.
With digestion, only tag arrays appear.

\begin{table}[t]
\centering
\caption{Configurations on the time--space--accuracy Pareto surface (time per read in $\mu$s, by threshold $L$), with proportional credit.
EF stands for Elias--Fano, gramm.\ runs for tags with grammar-compressed runs (Section~\ref{sec:tags}), $s$ is the sampling parameter of the sr-index or of the tags, and $t$ is the length of the strings in the lookup table: $t = 2$ digest symbols with digestion and, unless stated otherwise, $t = 10$ bases without.
The accuracy depends only on the digestion and on $L$ (strict, mean over the four regions).}
\label{tab:pareto}
\small
\begin{tabular}{lrrrr}
\toprule
configuration & space (GB) & $L = 30$ & $L = 40$ & $L = 50$ \\
\midrule
\multicolumn{2}{l}{\emph{minimizer digests}: strict accuracy} & 92.48\% & 92.36\% & 92.05\% \\
\quad RLBWT + tags, $s = 16$ & 0.16 & 735.6 & 361.4 & 225.1 \\
\quad RLBWT + tags, $s = 4$, no table & 0.18 & 240.1 & 162.0 & 125.6 \\
\quad RLBWT + tags, $s = 4$ & 0.19 & 220.1 & 147.0 & 114.3 \\
\quad EF RLCSA + tags & 0.27 & 61.9 & 50.7 & 44.4 \\
\quad explicit RLCSA + tags & 0.49 & 39.3 & 32.7 & 25.2 \\
\midrule
\multicolumn{2}{l}{\emph{undigested}: strict accuracy} & 93.77\% & 93.48\% & \\
\quad RLBWT + sr-index, $s = 16$, no table & 0.41 & 476.9 & 389.5 & \\
\quad RLBWT + sr-index, $s = 16$ & 0.46 & 315.1 & 265.9 & \\
\quad RLBWT + tags, $s = 4$ & 0.61 & 227.8 & 197.6 & \\
\quad RLBWT + gramm.\ runs & 0.64 & 211.5 & 186.6 & \\
\quad RLBWT + tags, $s = 4$, $t = 12$ & 0.88 & 216.2 & 185.8 & \\
\quad RLBWT + gramm.\ runs, $t = 12$ & 0.90 & 199.1 & 177.3 & \\
\quad explicit RLCSA + sr-index, $s = 16$, no table & 0.98 & 193.9 & 145.3 & \\
\quad explicit RLCSA + gramm.\ runs, no table & 1.02 & 104.0 & 85.0 & \\
\quad explicit RLCSA + gramm.\ runs & 1.04 & 75.0 & 62.4 & \\
\quad explicit RLCSA + gramm.\ runs, $t = 12$ & 1.31 & 71.0 & 60.0 & \\
\quad explicit RLCSA + tags (default) & 1.44 & 66.0 & 57.3 & \\
\quad explicit RLCSA + tags, $t = 12$ & 1.71 & 62.9 & 54.4 & \\
\bottomrule
\end{tabular}
\end{table}

\subsection{Counting with rank queries}
\label{app:rank}

Instead of scanning the tag runs in a MEM's BWT interval to count its occurrences in each taxon, we could list the taxa with Muthukrishnan's algorithm and then count the occurrences of each with two rank queries on the tag array, as Sadakane's document listing allows~\cite{Sad07} and as Cobas, M\"akinen and Rossi do on a grammar-compressed document array~\cite{CMR20}.
This takes time proportional to the number of taxa containing the MEM, times the cost of a rank query, rather than to the number of tag runs in the interval.
It needs rank support on the run-length compressed tag array, for example, for each taxon, the sorted list of the starts of its tag runs with their cumulative lengths, like the lists of our RLCSA, so that a rank query is a predecessor search in one list.
It can pay off only when intervals span many tag runs of few taxa, as the intervals of short or common MEMs can on references with many genomes per taxon.
We implemented it, with explicit lists and jump tables as in our RLCSA, and compared it with scanning on the BWT intervals of the MEMs of our SILVA timing reads and of 1\% of the Zymo reads; both methods give identical counts.
On SILVA, with $L = 30$, it takes 1.32\,$\mu$s per MEM against 0.92\,$\mu$s for scanning, and its lists take 3.0\,GB against 0.97\,GB for the whole tag array: intervals with many tag runs also contain many genera, so listing alone costs almost as much as scanning.
On Zymo it wins for the rare intervals with 64 or more tag runs (1.8\,$\mu$s against up to 10\,$\mu$s), which contain many runs of few species, but these are fewer than 0.03\% of the intervals, and 99\% lie inside a single tag run, where counting is trivial; its lists take 7.3\,GB against 1.47\,GB.
On both references, counting is a small fraction of the time per read.

\subsection{LCAs without locating: the rz-index}
\label{sec:rz}

When each match votes for an LCA, we only need the leftmost and rightmost genera containing it, which are the genera of its leftmost and rightmost occurrences in the text.
The rz-index finds these without locating any other occurrence.
It combines the first KATKA~\cite{GKL22}, which found the leftmost and rightmost occurrences of a $k$-mer with range-minimum queries on grids built from LZ77 parses, with the composite indexes of Belazzougui et al.~\cite{BCG+15,BCG+17}, which pair the RLBWT with the phrase boundaries of the LZ77 parse: an occurrence that crosses or ends at a phrase boundary (a \emph{primary} occurrence) is found with a two-dimensional grid whose points are the phrase boundaries, sorted by the text that follows them and by the reversed text that precedes them, and all other occurrences are copies of primary ones.
To query the grid for a pattern, they need the BWT interval of each of its suffixes, which backward search over the RLBWT of the text provides, and the interval of the reverse of each of its prefixes, for which they keep a second RLBWT, of the reversed text.
Since our index contains both strands, we do not need that second RLBWT: we sort the phrase boundaries by the reverse complement of the text preceding them, instead of its reverse, so the intervals we need are those of the reverse complements of the prefixes, which are in the same index.
Since the leftmost occurrence of a pattern is always primary, we can also replace range reporting with range-minimum queries on the grid.
A symmetric grid built on a right-to-left parse gives the rightmost occurrence.
For the rz-index, $S$ and its reverse are parsed with a left-referencing greedy LZ77 parse in which each phrase's source starts within the previous 32\,million characters.
The two grids are built on the phrase ends of the two parses, each as a wavelet matrix with a range-minimum structure on each level, and auxiliary structures (per-level union bitvectors and scan arrays) speed up the grid queries.

Table~\ref{tab:rz} compares the rz-index with Cliffy's LCA query.
With Cliffy's own matches (ZMPs), the rz-index is 4.4 to 9 times smaller than Cliffy's index, and on our older desktop it takes 1.6 to 3 times as long per read as Cliffy reports on its server (1.6 times with the explicit RLCSA).
Its accuracy is within a few points of Cliffy's on a different set of reads: 1 point higher on V1--V2 and 1 to 4 points lower on the other three regions, which suggests that the two read sets differ somewhat in difficulty, and which served as a check on our setup.
With MEMs found by BML and $L = 50$, the digested rz-index classifies a read in 66\,$\mu$s, faster than Cliffy reports for its digested LCA query on its server, with an index 4.4 times smaller and similar accuracy.

\begin{table}[t]
\centering
\caption{LCA classification: the rz-index with ZMPs (Cliffy's matches; scored on 1 million read pairs per region) and with BML MEMs, against Cliffy's LCA query.
Accuracy is at the genus level, with Cliffy's scoring (``theirs'') and strict.
Ranges are over the four regions and single values are means, and we did not compute strict accuracy for the ZMP runs.
The rz-index sizes include auxiliary structures that speed up its grid queries.
EF stands for Elias--Fano.
Cliffy's time is for approximate listing, which uses the same index as its LCA query.}
\label{tab:rz}
\small
\resizebox{\textwidth}{!}{\begin{tabular}{llrrcc}
\toprule
 & configuration & space (GB) & time ($\mu$s/read) & theirs & strict \\
\midrule
\multirow{3}{*}{\shortstack[l]{digested\\ZMPs}} & rz, RLBWT & 0.76 & 212 & \multirow{3}{*}{83.7--87.3\%} & \\
 & rz, EF RLCSA & 0.75 & 152 & & \\
 & rz, explicit RLCSA & 0.96 & 125 & & \\
\addlinespace
\multirow{3}{*}{\shortstack[l]{undigested\\ZMPs}} & rz, RLBWT & 0.99 & 522 & \multirow{3}{*}{84.8--88.6\%} & \\
 & rz, EF RLCSA & 1.01 & 594 & & \\
 & rz, explicit RLCSA & 1.39 & 330 & & \\
\addlinespace
digested MEMs & rz, explicit RLCSA, $L = 30$ & 0.96 & 131 & 85.7\% & 85.3\% \\
 & rz, explicit RLCSA, $L = 50$ & 0.96 & 66 & 86.7\% & 84.8\% \\
undigested MEMs & rz, explicit RLCSA, $L = 30$ & 1.39 & 262 & 87.6\% & 87.0\% \\
 & rz, explicit RLCSA, $L = 50$ & 1.39 & 167 & 89.3\% & 85.5\% \\
\midrule
Cliffy LCA & minimizers & 4.27 & 77--82 & 86.3--88.5\% & \\
 & undigested & 9.01 & 199--204 & 87.5--90.2\% & \\
\bottomrule
\end{tabular}}
\end{table}

The rz-index is not on our Pareto surface, however, because with tag arrays listing costs little more than computing the LCA and is about 3.5 points more accurate with even credit, and more with proportional credit.
Computing the LCA from tags, as the smallest and largest genus listed, also gives exactly the same answers as the rz-index, about three times faster: over digests, with $L = 30$ and an explicit RLCSA, about 40\,$\mu$s per read against 131\,$\mu$s.
The rz-index may still be useful where listing is not an option, for example when the number of documents is so large that listing them is too expensive.

\subsection{Verifying digested matches}

Undigested MEMs are more accurate than digested ones, so we tried to get the best of both by verifying each occurrence of a digested MEM against the DNA and extending it to an exact match, using a Relative Lempel--Ziv representation of the database (with the longest sequence of each genus as the reference) and the positions of the minimizers.
These structures took 0.24\,GB, and verification worked, but digested MEMs are less specific than DNA MEMs (with $L = 30$, they occur in 6.0 genera on average, against 3.8), and checking their occurrences took 850\,$\mu$s per read, against 125 to 253\,$\mu$s for classifying with an undigested sr-index on an explicit RLCSA with a lookup table ($L$ from 30 down to 15).

\subsection{Combining LCAs and listing}

Listing is most expensive for MEMs with many occurrences, and for such MEMs the LCA may lose little information, so we tried answering by LCA (with the rz-index's grids) whenever a MEM has more than $K$ occurrences, or whenever its occurrences cover a fraction at least $\alpha$ of the sequences under its LCA, and listing (with the sr-index) otherwise.
With BML, few MEMs have many occurrences: the threshold $K$ saved 25 to 45\% of the time over digests, but required the grids (0.42\,GB more), and the rule based on $\alpha$ sent fewer than 1\% of the MEMs to the LCA while needing the LCA of every MEM to decide.
With tag arrays, listing is as fast as computing the LCA anyway.

\subsection{Filtering \texorpdfstring{$k$}{k}-mers instead of digesting}

KeBaB~\cite{BDZ+25} first breaks each read into \emph{pseudo-MEMs}, the maximal substrings of length at least $L$ all of whose $k$-mers pass a Bloom filter of the database's $k$-mers ($k < L$).
Every MEM of length at least $L$ lies in a pseudo-MEM, so the rest of the read can be skipped without changing the answers.
Since this keeps the accuracy of undigested matches, we asked whether it could give us the speed of digests instead.
To bound what any filter could achieve, we used exact $k$-mer membership at no cost before undigested BML (listing with tags on the explicit RLCSA).
The answers were identical, but 94 to 97\% of the bases remained in pseudo-MEMs, and the time per read dropped only from 92--99 to 72--83\,$\mu$s for $L = 30$ ($k = 20$ or $25$) and by 7 to 9\% for $L = 15$ ($k = 12$), against 42\,$\mu$s with digests and an index a third of the size (all without lookup tables).
Almost every $k$-mer of our short, accurate reads occurs in the database, and BML spends its time extending the long MEMs themselves.
KeBaB is designed for long, error-prone reads.
Brown et al.\ obtained their largest speedups by searching only each read's longest pseudo-MEMs, which risks discarding all of its long MEMs, since a long pseudo-MEM need not contain a long MEM: in their pipeline the pseudo-MEMs were chosen before any was searched.
If finding pseudo-MEMs is integrated with searching them, this risk disappears: we can search the pseudo-MEMs in decreasing order of length and stop only when we have found as many MEMs as we want that are longer than the next pseudo-MEM~\cite{BDZ+25}, or use parse indexing to bound the lengths of the MEMs that pseudo-MEMs must contain~\cite{Gag26}.

\subsection{Matching whole phrases}

A digest symbol names only its minimizer, not the bases up to the next one.
A one-level parse index~\cite{HOK+24,BGM+25,Gag26} instead replaces each phrase of a locally consistent parse, such as prefix-free parsing~\cite{BGK+19} or a parse into phrases starting at minimizers, by an identifier, such as its Karp--Rabin hash, and parses the read the same way, so that matching whole phrases is exact, at about the cost of matching digests.
Only the ends of each MEM, up to about a phrase on each side, are lost.
To estimate that loss, we trimmed each MEM found with $L = 30$ by the undigested index to the part a phrase index with minimizer phrases ($k = 4$, $w = 11$) would certainly match (18.5 bases removed per MEM, on average) and listed the genera of the trimmed strings.
On 20\,000 read pairs from each of V1--V2 and V4, strict accuracy dropped by 2.3 and 3.3 points, more than digestion costs, and trimming just 2 bases from each end already cost 0.3 to 0.5 points: the ends of a MEM are where the read diverges from all but a few genera, and the middles of our MEMs occur in 219 genera on average, against 4 for the whole MEMs.
A one-level phrase index is therefore no substitute for digests on SILVA, although it might be on datasets whose matches carry less of their information at their ends.

\end{document}